\documentclass[11pt]{article}

\usepackage{graphicx,amssymb,amsmath}
\usepackage{epsfig}
\usepackage{epsf}
\usepackage{datetime}

\usepackage[small]{caption}
\usepackage{amsfonts}
\usepackage{latexsym}
\usepackage{fancyvrb}
\usepackage{amssymb}
\usepackage{amsthm}
\usepackage{graphics}
\usepackage{psfrag}
\usepackage{bbm}
\usepackage{url}

\newtheorem{theorem}{Theorem}
\newtheorem{lemma}{Lemma}

\newtheorem{observation}{Observation}

\newcommand{\old}[1]{{}}
\newcommand{\later}[1]{{}}

\def\etal{{et~al.}}
\def\eg{{e.g.}}
\def\ie{{i.e.}}

\newcommand{\eps}{\varepsilon}

\newcommand{\NN}{\mathbb{N}}
\newcommand{\RR}{\mathbb{R}}

\def\D{\mathcal D}

\def\H{\mathcal H}
\def\I{\mathcal I}
\def\L{\mathcal L}

\def\R{\mathcal R}

\newcommand{\alg}{{\rm ALG}}
\newcommand{\opt}{{\rm OPT}}

\newcommand{\len}{{\rm len}}
\newcommand{\per}{{\rm per}}
\newcommand{\diam}{{\rm diam}}
\newcommand{\conv}{{\rm conv}}
\newcommand{\area}{{\rm Area}}
\newcommand{\vol}{{\rm Vol}}
\newcommand{\proj}{{\rm proj}}

\title{The Traveling Salesman Problem for Lines, Balls and
  Planes\thanks{A preliminary version has appeared in the Proceedings of
    the 24th ACM-SIAM Symposium on Discrete Algorithms,
    New Orleans, LA, 2013, SIAM, pp.~828--843.}}

\author{Adrian Dumitrescu\thanks{Department of Computer Science,
University of Wisconsin--Milwaukee, WI, USA\@.
Email:~\texttt{dumitres@uwm.edu}.
Research by this author was supported in part by the NSF grant DMS-1001667.}
\qquad
Csaba D. T\'oth\thanks{Department of Mathematics, California State
  University, Northridge, Los Angeles, CA;
and Department of Computer Science, Tufts University, Medford, MA, USA\@.
Email:~\texttt{cdtoth@acm.org}.
Research by this author was supported in part by NSERC (RGPIN 35586)
and NSF (CCF-0830734 and CCF-1423615).}}

\begin{document}

\maketitle

\begin{abstract}
We revisit the traveling salesman problem with neighborhoods
(TSPN) and propose several new approximation algorithms. These constitute either
first approximations (for hyperplanes, lines, and balls in $\RR^d$, for $d\geq 3$)
or improvements over previous approximations achievable in comparable times
(for unit disks in the plane).

\smallskip
(I) Given a set of $n$ hyperplanes in $\RR^d$,
a TSP tour whose length is at most $O(1)$ times the optimal can be computed in
$O(n)$ time, when $d$ is constant.

\smallskip
(II) Given a set of $n$ lines in $\RR^d$, a TSP tour whose length is at most
$O(\log^3 n)$ times the optimal can be computed in polynomial time for all $d$.

\smallskip
(III) Given a set of $n$ unit balls in $\RR^d$, a TSP tour whose length is at most
$O(1)$ times the optimal can be computed in polynomial time, when $d$ is constant.

\medskip
\noindent\textbf{\small Keywords}:
Traveling salesman,
group Steiner tree,
linear programming,
minimum-perimeter rectangular box,
approximation algorithm,
lines, planes,
hyperplanes,
unit disks and balls.
\end{abstract}

\section{Introduction}  \label{sec:intro}

In the Euclidean Traveling Salesman Problem (ETSP), given a set
of points in the plane (or in the Euclidean space $\RR^d$, $d \geq 3$),
one seeks a shortest tour (closed curve) that visits each point.
In the \emph{TSP with neighborhoods} (TSPN), first studied by
Arkin and Hassin~\cite{AH94}, each point is replaced by a
(possibly disconnected) region. The tour must visit
at least one point in each of the given regions (\ie, it
must intersect each region).
A~tour for a set of neighborhoods is also referred to as a TSP tour.
Since ETSP is known to be NP-hard in $\RR^d$
for every $d \geq 2$ \cite{GGJ76,GJ79,P77}, TSPN is also NP-hard
for every $d \geq 2$. TSP is recognized as one of the corner-stone
problems in combinatorial optimization. See~\cite{Mi00,Mi04} for a
list of related problems in geometric network optimization.

\vspace{-5pt}
\paragraph{Related work.}
It is known that ETSP admits a polynomial-time
approximation scheme in $\RR^d$, where $d=O(1)$,
due to Arora~\cite{Ar98} and Mitchell~\cite{Mi99}.
Subsequent running time improvements were obtained
by Rao and Smith~\cite{RS98}; specifically, the running time
of their PTAS is $O(f(\eps) \, n \log{n})$, where $f(\eps)$ grows
exponentially in $1/\eps$. In contrast, TSPN in general is harder to approximate.
Certain instances are known to be APX-hard. Research efforts focused on
approximations for families of neighborhoods with ``nice'' geometric properties.
Typically, improved approximation methods are available when the neighborhoods are
pairwise disjoint, or fat, or have comparable sizes. We briefly review previous
work most closely related to our results.

Arkin and Hassin~\cite{AH94} gave constant-factor approximations for
translates of a convex region, translates of a connected region, and
more generally, for regions with diameters \emph{parallel} to a common direction
and of comparable length (within a constant factor).
Dumitrescu and Mitchell~\cite{DM03} extended the above result to
arbitrary connected neighborhoods with comparable diameters.

For $n$ connected (possibly overlapping) neighborhoods in the plane,
TSPN can be approximated with ratio $O(\log{n})$ by the algorithm of
Mata and Mitchell~\cite{MM95}. See also the survey by Bern and
Eppstein~\cite{BE97} for a short outline of this algorithm.
Subsequent running time improvements were offered
by Elbassioni~\etal~\cite{EFS06} and by Gudmundsson and
Levcopoulos~\cite{GL99}. At its core, the $O(\log{n})$-approximation
relies on the following early result by Levcopoulos and Lingas~\cite{LL84}:
Every (simple) rectilinear polygon $P$ with $n$ vertices, $r$ of which are reflex,
can be partitioned in $O(n \log{n})$ time into rectangles whose total perimeter
is $\log{r}$ times the perimeter of $P$.

Bodlaender~\etal~\cite{BFG+09} gave a PTAS for TSPN for disjoint fat
regions of about the same size (this includes the case of disjoint unit disks)
in $\RR^d$, where $d$ is constant. Earlier Dumitrescu and Mitchell~\cite{DM03}
proposed a PTAS for TSPN for fat regions of about the same size and bounded
depth in the plane, where Spirkl~\cite{Sp14} recently found and filled a gap.

Using an approximation algorithm due to Slavik~\cite{Sla97} for Euclidean group
TSP (see below), de Berg~\etal~\cite{BGK+05} obtained constant-factor approximations
for disjoint fat convex regions in the plane, not necessarily of comparable size.
Elbassioni~\etal~\cite{EFMS05} improved the runtime of the approximation algorithm.
Subsequently, Elbassioni~\etal~\cite{EFS06,EFS09} gave constant-factor
approximations for (possibly intersecting) fat convex regions of comparable size.
Preliminary work by Mitchell gave
(i) a PTAS~\cite{Mi07} for bounded depth fat regions of arbitrary sizes
in the plane; in particular for disjoint fat regions in the plane, and
(ii) constant-factor approximations for pairwise-disjoint connected
neighborhoods of any size or shape~\cite{Mi10}.
Very recently, Chan and Jiang~\cite{CJ16} gave a PTAS for
\emph{fat weakly disjoint regions} in metric spaces of constant doubling
dimension by combining a QPTAS by Chan and Elbassioni~\cite{CE11}
with a PTAS for TSP in doubling metrics by Bartal~\etal~\cite{BGK12}.
(For example, disjoint unit balls in $\RR^d$, $d\geq 2$, are fat weakly
disjoint regions per the definition in \cite{CJ16}, but disjoint balls of arbitrary
radii need not be).
A constant-factor approximation for disks in the Euclidean plane
(with arbitrary radii and overlaps) was obtained in~\cite{DT15}.

Finally, interesting variants are those with unbounded neighborhoods,
such as lines or planes. For TSPN for $n$ lines in the plane, an
exact solution can be found in $O(n^5)$
time~\cite{CJN99,DELM03,THI99,Ta01} (see also~\cite{J02}),
and a $1.28$-approximation can be computed in $O(n)$ time~\cite{Du12}.
In contrast, TSPN for lines in $\RR^3$ is NP-hard.
The status of TSPN for planes in $\RR^3$ appears to be unknown.

Regarding the degree of approximation achievable, TSPN for arbitrary
neighborhoods is generally APX-hard~\cite{BGK+05,SS05}, and it remains
so even for segments of nearly the same length~\cite{EFS06}.
For instance, approximating TSPN for connected regions in
the plane within a factor smaller than 2 is intractable (NP-hard)~\cite{SS05}.
The problem is also APX-hard for disconnected regions~\cite{SS05},
the simplest case being point-pair regions~\cite{DO08}.
It is conjectured that approximating TSPN for disconnected regions in
the plane within a $O(\log^{1/2} n)$ factor is intractable~\cite{SS05}.
Similarly, it is conjectured that approximating TSPN for connected
regions in $\RR^3$ within a $O(\log^{1/2} n)$ factor and for disconnected
regions in $\RR^3$ within a $O(\log^{2/3} n)$ factor~\cite{SS05} are
intractable. Moreover, proving these conjectures seems to require advances in
complexity, rather than geometry.

\vspace{-5pt}
\paragraph{Our results.}
In this paper we present several improved approximation
algorithms for TSPN, for three types of neighborhoods:
(i) hyperplanes in $\RR^d$;
(ii) lines in $\RR^d$;
(iii) congruent disks in the plane and congruent balls in $\RR^d$.
Our results and related older results are summarized in Table~\ref{table}.
\medskip
\begin{table*}[t]
\begin{center}
\begin{tabular}{|l|c|c|c|c|c||} \hline
&Region type & Old ratio & New ratio & NP-hard \\ \hline \hline
1&Hyperplanes in $\RR^d$, $d\geq 3$ & --- & $(1+\eps) \, 2^{d-1}/\sqrt{d} $ & open \\ \hline
2&Planes in $\RR^3$ & --- & $2.31$ in $O(n)$ time & open \\ \hline
3&Lines in $\RR^d$, $d\geq 3$ & --- & $O(\log^3 n)$ & yes \\ \hline
4&Disjoint unit disks in the plane & $3.55$ & --- & yes \\ \hline
5&Unit disks in the plane & $7.62$ & $6.75$ & yes \\ \hline
6&Disjoint unit balls in $\RR^3$ & --- & $7.01 $ & yes \\ \hline
7&Unit balls in $\RR^3$ & --- &  $100.61$ & yes \\ \hline
8&Unit balls in $\RR^d$ & --- &  $ O(7.73^d) $ & yes \\ \hline
9&Disjoint balls in $\RR^d$ &  $O(2^d/\sqrt{d})$  &  --- & yes \\ \hline
\end{tabular}
\end{center}
\caption{Old and new (asymptotic) approximation ratios obtained in polynomial time.
The ratios in rows 4--8 are obtained by using a black box PTAS for computing point tours.
Disjoint unit balls in $\RR^d$, $d\geq 2$, admit a PTAS~\cite{BFG+09,CJ16,DM03,Sp14}.
The old ratios listed in column 2 are from~\cite{DM03} (rows 4,5) and~\cite{EFS09} (row 9).}
\label{table}
\end{table*}

We start with hyperplanes in $\RR^d$; no approximation algorithm
was known for this type of neighborhoods. For constant $d$, we can
compute constant-factor approximations in linear time.

\begin{theorem}\label{thm:planes}
Given a set of $n$ hyperplanes in $\RR^d$, and $\eps>0$,
a TSP tour  whose length is at most $(1+\eps) \, 2^{d-1}/\sqrt{d}$ times
the optimal can be computed in at most $O(C_{d,\eps}\ n)$ time,
where $C_{d,\eps} = d^2 2^{2d} \, (d/\eps)^d$.
In particular for $d=3$, a TSP tour whose length is at most
$2.31$ times the optimal can be computed in $O(n)$ time.
\end{theorem}

We continue with lines in $\RR^d$, a problem much harder to deal with.
Note that an instance with parallel lines reduces to an
instance of ETSP for points in one dimension lower (namely
the points of intersection between the given lines orthogonal to a hyperplane).
Here we obtain the first approximations.

\begin{theorem}\label{thm:lines}
  Given a set of $n$ lines in $\RR^d$, a TSP tour  whose length is at most
  $O(\log^3 n)$ times the optimal can be computed in time $O(d \cdot {\rm poly}(n))$.
\end{theorem}

While for \emph{disjoint} unit balls in $\RR^d$, $d\geq 2$,
the existence of a PTAS has been established~\cite{BFG+09,CJ16,DM03,Sp14},
no PTAS is known for intersecting unit balls in any dimension $d\geq 2$.
For arbitrary unit balls in $\RR^d$, we give constant-factor approximations
by using a black box that computes a good tour of at most $n$ points
(the centers of a suitable subset of disks, resp., balls).
For unit disks in $\RR^2$, we obtain an improved approximation factor $6.75$;
the previous best ratio, $7.62$, holds for \emph{translates}
of a convex region~\cite{DM03}. Let $T(n,d,\eps)$ denote the running
time for computing a $(1+\eps)$-approximation of an optimal tour of
$n$ points in $\RR^d$; recall that $T(n,d,\eps)$ is currently
exponential in $1/\eps$~\cite{RS98}.

\begin{theorem}\label{thm:3}
Given a set of $n$ unit disks in the plane, and $\eps>0$, a TSP tour whose length
is at most $\left(\frac{7}{3} + \frac{8\sqrt3}{\pi}\right) (1+ \eps)$
times the optimal, apart from an additive constant,
can be computed in time $O(T(n, 2, 1.8 \, \eps))$.
In particular, a TSP tour whose length is at most $6.75$ times the
optimal can be computed in time $O(T(n,2,0.0018))$.
Alternatively, a TSP tour  whose length is at most $8.52$ times the
optimal can be computed in time $O(n^{3/2} \log^5{n})$.
\end{theorem}

For congruent balls in $\RR^3$ we give the first explicit
constant approximation factor, \emph{not} in the $O(1)$ form.

\begin{theorem}\label{thm:4}
Given a set of $n$ unit balls in $\RR^3$, and $\eps>0$, a TSP tour  whose length
is at most $54 \sqrt3 (1+\eps)$ times the optimal,
apart from an additive constant, can be computed in time $O(T(n,3,\eps))$.
In particular, a TSP tour whose length is at most $100.61$ times the optimal
can be computed in time $O(T(n,3,0.01))$.
Alternatively, a TSP tour  whose length is at most $104.1$ times the optimal
can be computed in time $O(n^{3})$.
\end{theorem}

The above result generalizes to congruent balls in $\RR^d$ for any fixed
dimension $d$; the proof is analogous to that of Theorem~\ref{thm:4} for
the $3$-dimensional version.

\begin{theorem}\label{thm:5}
Given a set of $n$ unit balls in $\RR^d$, and $\eps>0$, a TSP tour whose length
is at most $O(7.73^d)$ times the optimal can be computed in time $O(T(n,d,\eps))$.
\end{theorem}

\paragraph{Preliminaries.}
Let $\R$ be a set of regions in $\RR^d$, $d \geq 2$.
A set $\Xi \subset \RR^d$ \emph{intersects} $\R$ if
$\Xi$ intersects each region in $\R$, that is, $\Xi \cap r \neq
\emptyset$, $\forall r \in \R$. A shortest TSP tour for a set $\R$
of regions (neighborhoods), denoted by $\opt(\R)$, is a shortest
closed curve in the ambient space that intersects~$\R$.

The Euclidean length of a curve $\gamma$ is denoted by $\len(\gamma)$,
or just $|\gamma|$ when there is no danger of confusion.
Similarly, the total (Euclidean) length of the edges
of a geometric graph $G$ or a polygon $P$ is denoted by $\len(G)$ and
$\per(P)$, respectively. For a hyperrectangle (rectangular box) $Q$ in
$\RR^d$ with sides $w_1,\ldots,w_d$, the total edge length
$\per(Q)=2^{d-1}\sum_{i=1}^d w_i$ is called its \emph{perimeter}.

For $\alpha \geq 1$, we say that an approximation algorithm (for TSPN)
has ratio $\alpha$ if its output tour $\alg$ satisfies $\len(\alg)
\leq \alpha \, \len(\opt)$, where $\opt$ is an optimal tour,
and has \emph{asymptotic} ratio $\alpha$ if its output satisfies
$\len(\alg) \leq \alpha \, \len(\opt) + \beta$ for some constant $\beta \geq 0$.

The convex hull of a set $A\subset \RR^d$ is denoted
by $\conv(A)$. The Cartesian coordinates of a point
$p \in \RR^d$ are denoted by $x_1(p),\ldots , x_d(p)$.
For a line segment $s \in \RR^3$, $\Delta_1(s),\ldots ,\Delta_d(s)$
denote the lengths of its projections on the $d$ coordinate axes.

\section{The illusions and  pitfalls of localization} \label{sec:localization}

Given a set $\R$ of $n$ regions, it would be helpful to find a convex
set that contains an optimal tour $\opt=\opt(\R)$ and whose diameter
is a polynomial in $n$ and perhaps other parameters, such as an upper
bound on $\diam(\opt)$.
A convex set $C_1$ that intersects $\R$ is often easy to compute. It
is tempting to believe (as it has been suggested by several researchers)
that if $C_1$ is scaled up by some suitable polynomial factor, the
resulting convex set $C_2$ might contain $\opt$. Finding such a set
$C_2$ would allow using standard approximation techniques (such as
discretization, convex approximation tools, etc.).

In this section, we show that this na\"{\i}ve approach is infeasible
when the regions in $\R$ are lines or hyperplanes in $\RR^d$.
Let $\lambda(x,y)$ be a given polynomial of $2$ variables with
positive coefficients.
We present constructions for a set of $n$ lines and a set of $n$ hyperplanes,
respectively, such that the minimum intersecting ball $B_1$ is
centered at the origin, but $\lambda \, B_1$ fails to contain $\opt$,
where $\lambda=\lambda(n,\diam(B_1))$. Moreover:
(i) the shortest TSP tour contained in $\lambda \, B_1$ is a
$\Theta(\sqrt{n})$-approximation for lines, in contrast with the
$O(\log^3n)$-approximation in Theorem~\ref{thm:lines}, and
(ii) the shortest TSP tour contained in $\lambda \, B_1$ is a
$c$-approximation for hyperplanes, where $c>1$ is a constant, which
rules out a $(1+\eps)$-approximation algorithm using this approach.
\begin{figure}[htb]
\centerline{\epsfxsize=5in \epsffile{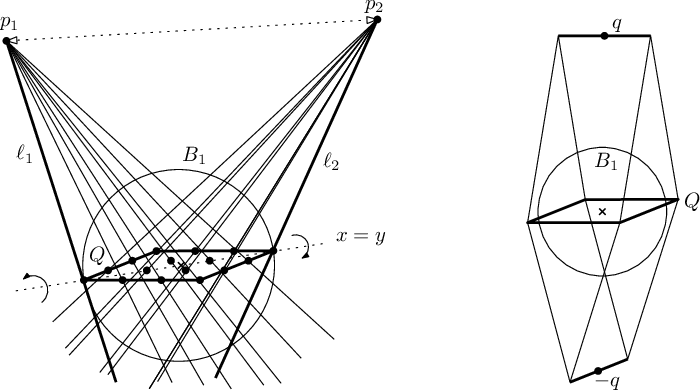}}
\caption{Left: a set $\L$ of nearly vertical lines that intersect a square $Q$
in a grid-like pattern, and their minimum intersecting ball $B_1$.
Right: a set of four nearly vertical planes containing
four sides of a square $Q=[-1,1]^2$ in the $xy$-plane,
and their minimum intersecting ball $B_1$.}
\label{fig:pitfalls}
\end{figure}

\paragraph{Lines in $\RR^3$.}
For an integer $n$ and a polynomial $\lambda(x,y)$, we construct a set
$\L$ of $n$ lines in $\RR^3$. Consider the square $Q=[-1,1]^2$ in the
$xy$-plane (Fig.~\ref{fig:pitfalls}~(left)). Let $B_1$ be the
ball of radius $\sqrt{2}$ centered at the origin, and note that
$Q \subset B_1$. We first construct two skew lines in $\RR^3$ whose
minimum intersecting ball is $B_1$. Start with two vertical lines
passing through $(1,1,0)$ and $(-1,-1,0)$, and observe that they
intersect any horizontal plane at two points at distance $2\sqrt{2}$
apart. Rotate these lines about the horizontal line $\ell_0: y=x$ by
some small angle $\alpha$ and $-\alpha$, respectively, to obtain two skew
lines $\ell_1$ and $\ell_2$. As $\ell_1$ and $\ell_2$ remain orthogonal
to $\ell_0$, the minimum intersecting ball of $\ell_1$ and $\ell_2$ is
still $B_1$. Choose $\alpha$ such that $\ell_1$ and $\ell_2$ intersect
the horizontal plane $z=n \, \lambda(n,4)$ at two points, $p_1$ and $p_2$,
at distance $4$ apart. We now define the set $\L$ of $n$ lines as follows:
$\L$ contains $\ell_1$ and $\ell_2$,
about half of the lines in $\L$ pass through $p_1$ and the
other half pass through $p_2$. The lines in $\L$ are nearly vertical
and intersect $Q$ in a square grid pattern, where any two intersection
points are at distance at least $2/\sqrt{n}$ apart.

\begin{lemma}\label{lem:pitfalls1}
Every TSP tour $\gamma$ lying in $\lambda\, B_1$ satisfies
$\len(\gamma) \geq \frac{\sqrt{n}}{8}\, \len(\opt)$.
In particular, $\lambda\,  B_1$ does not contain the optimal tour
$\opt=\opt(\L)$ or any $o(\sqrt{n})$-approximation of it.
\end{lemma}
\begin{proof}
Note that the tour that visits points $p_1$ and $p_2$, of length
$2|p_1p_2|=8$, intersects all lines. Consequently, $\len(\opt) \leq 8$
and $\diam(\opt) \leq 4$.
Consider a tour $\gamma$ lying in $\lambda\, B_1$
and let $\gamma'$ be the orthogonal projection of $\gamma$ onto the
$xy$-plane, where $\len(\gamma')\leq \len(\gamma)$.
Since the lines in $\L$ are nearly vertical, the orthogonal
projections of the line segments in $\{\ell\cap \lambda\, B_1: \ell\in
\L\}$ have length at most $2/n$, and they each contain distinct grid
points within $Q$. Since the distance between any two grid points
is at least $2/\sqrt{n}$, we have $\len(\gamma')\geq
n(2/\sqrt{n}-4/n)= 2\sqrt{n}-4\geq \sqrt{n}$, and so
$\len(\gamma)\geq \sqrt{n} \geq \frac{\sqrt{n}}{8}\, \len(\opt)$,
as required.
\end{proof}

\paragraph{Planes in $\RR^3$.}
For an integer $n$ and a polynomial $\lambda(x,y)$, we construct a set
$\H$ of $n$ planes in $\RR^3$.
Consider the unit square $Q=[-1,1]^2$ in the $xy$-plane
(Fig.~\ref{fig:pitfalls}~(right)). Let the first 4 planes in $\H$ each
contain one side of $Q$. The two planes containing the two sides of $Q$
parallel to the $x$-axis intersect in a line parallel to the $x$-axis
and containing the point $q=(0,0,h)$, where $h$ is large, specifically
$h=n \, \lambda(n,3)$. The two planes containing the sides of $Q$
parallel to the $y$-axis intersect in a line parallel to the $y$-axis
and containing the point $-q=(0,0,-h)$.
By symmetry, the minimum intersecting ball of these four planes is
centered at the origin, and its radius is at least $1-1/h$ and at most
$1$. Arrange the remaining $n-4$ planes in $\H$ such that
they all contain the point $q=(0,0,h)$, are tangent to the ball $B_1$,
and the tangency points are uniformly distributed along a horizontal
circle $C \subset \partial B_1$. By construction, $B_1$ is the minimum
intersecting ball of the $n$ planes in $\H$.

\begin{lemma}\label{lem:pitfalls2}
Every TSP tour $\gamma$ lying in $\lambda\, B_1$ satisfies
$\len(\gamma) \geq \frac{\pi}{2}(1-O(1/n))\, \len(\opt)$.
In particular, $\lambda\, B_1$ does not contain the optimal tour
$\opt=\opt(\H)$ or any $(1+\eps)$-approximation of it for a
sufficiently small $\eps>0$.
\end{lemma}
\begin{proof}
Note that the triangle formed by the point $q$ and its orthogonal
projections onto the two planes containing the two sides of $Q$ parallel
to the $y$-axis is a tour for $\H$. The length of this tour is at most
$4+4/h$. Consequently, $\len(\opt)\leq 4+4/h$ and $\diam(\opt) \leq 3$.
Consider a tour $\gamma$ lying in $\lambda\, B_1$,
and let $\gamma'$ be the orthogonal projection
of $\gamma$ to the $xy$-plane, where $\len(\gamma')\leq \len(\gamma)$.
Since the planes in $\H$ are nearly vertical, the orthogonal
projections of the disks in $\{H \cap \lambda\, B_1: H\in \H\}$ are
ellipses of width at most $2/n$. The first four ellipses each contain
a side of the square $Q$. The remaining ellipses form $\lfloor (n-4)/2\rfloor$
pairs such that the major axes of any pair are on parallel lines at
distance at least $2-2/n$ apart, and the directions of the pairs are
uniformly distributed.
Consequently, the width of $\gamma'$ is at least $2-O(1/n)$, and so
$\len(\gamma')\geq 2\pi(1-O(1/n))$ $\geq \frac{\pi}{2}(1-O(1/n))\, \len(\opt)$,
as required.
\end{proof}

\paragraph{Easy weak approximations.}
Finding a minimum-radius ball $B_1$ that intersects a set of $n$ hyperplanes
(resp., lines) in $\RR^d$ is an LP-type problem~\cite{DNW04};
for a fixed $d$, such a ball can be computed in $O(n)$ time.
This immediately leads to a simple $2^{d-1}$-approximation for
hyperplanes and a $O(n^{1-1/(d-2)})$-approximation for lines in
$\RR^d$. Indeed, since the minimum enclosing ball $B_\opt$ of an optimal
tour $\opt$ also intersects all $n$ hyperplanes (resp., lines), it
is clear that $\diam(B_1) \leq \diam(B_\opt)$.
Since $B_\opt$ is spanned by up to $d+1$ points, it is easy to see that
$\len(\opt) \geq 2 \, \diam(B_\opt)$. On the other hand,
a Hamiltonian cycle of the $2^d$ vertices of an enclosing hypercube of
$B_1$ intersects all hyperplanes (cf.~Observation~\ref{O1}), and has
length at most $2^d \diam(B_1)$.
For $n$ lines in $\RR^d$, one can compute
all intersection points of the $n$ lines with the boundary of $B_1$,
and return an approximate tour for these $2n$ points of length
$\diam(B_1) \cdot O(n^{1-1/(d-2)})$ by a result of Few~\cite{Fe55}.

In Section~\ref{sec:hyperplanes} we obtain a better approximation
for TSPN for $n$ hyperplanes, a ratio close to $2^{d-1}/\sqrt{d}$, by
using hyperrectangles instead of balls and a careful analysis.
In Section~\ref{sec:lines}, we use a completely different approach
to achieve a much better $O(\log^3 n)$-approximation for TSPN
for $n$ lines in $\RR^3$.

\section{TSPN for hyperplanes in $\RR^d$}  \label{sec:hyperplanes}

In this section we prove Theorem~\ref{thm:planes}:
we present a constant factor approximation algorithm
for TSPN for a set $\H$ of $n$ hyperplanes in $\RR^d$
with ratio $(1+\eps) \frac{2^{d-1}}{\sqrt{d}}$ and running in $O(n)$ time,
for constant $d$ and $\eps>0$. In particular, for $\eps =0.0002$,
we get the approximation ratios $2.31$ in $\RR^3$, $4.001$ in $\RR^4$,
and $7.16$ in $\RR^5$.

Our algorithm is based on solving low-dimensional linear programs;
it combines ideas from~\cite{Du12,DJ10,DM03,J02}. We show below
(Lemma~\ref{lem:curve}) that any closed curve $\gamma\subset \RR^d$
is contained in a rectangular box of edge lengths $w_1,\ldots , w_d$
such that $\sum_{i=1}^d w_i \leq \frac{\sqrt{d}}{2} \, \len(\gamma)$.
We apply this result to the optimal tour $\opt(\H)$.
Then we use linear programming to compute a $(1+\eps)$-approximation
for the minimum-perimeter rectangular box intersecting $\H$, and
produce a Hamiltonian cycle of the $2^d$ vertices as an approximate tour.

Let $Q$ be rectangular box in $\RR^d$ such that the $d$ extents of $Q$ are
$w_1\leq w_2\leq \ldots \leq w_d$. It is not difficult to see (by induction on $d$)
that $Q$ admits a Hamiltonian cycle of total length
$$\tau(Q) = 2^{d-1}w_1+2^{d-2}w_2+\ldots +2w_{d-1}+2w_d=w_d +\sum_{j=1}^d 2^{d-j}w_j .$$

The \emph{orientation} of a rectangular box $Q$ in $\RR^d$ is given by
an orthonormal basis whose vectors are parallel to the edges of $Q$.
Cover the unit sphere $\mathbb{S}^{d-1}\subset \RR^d$ with spherical caps of radius
$r=\eps/(d-1)$. Since the (spherical) volume of $\mathbb{S}^{d-1}$ is constant,
and the volume of a spherical cap of radius $r$ is $\Theta(r^{d-1})=\Theta((d/\eps)^{d-1})$,
we can select a set $A=\{\alpha_1,\ldots , \alpha_m\}$ of $m=O(d^d \eps^{1-d})$ orientations
that cover all possible orientations within an error of $\eps/(d-1)$. That is, for
any orientation $\alpha$, there is an orientation $\alpha'\in A$ and a matching
between the orthogonal bases $\alpha$ and $\alpha'$ so that the angle between
any two corresponding vectors is at most $\eps/(d-1)$.

\medskip
\noindent{\bf Algorithm~A1.}
\begin{itemize}
\item[] {\sc Step 1:} Let $m = O(d^d \eps^{1-d})$.
  For each $i=1,\ldots,m$,
  compute a minimum-perimeter rectangular box $Q_i$ with orientation
  $\alpha_i$ that intersects $\H$.

\item[] {\sc Step 2:} Let $Q$ be a box with the minimum perimeter over
  all $m$ directions, found above. Return a Hamiltonian cycle of the $2^d$
  vertices of $Q$, of length $\tau(Q)$, as depicted in Fig.~\ref{f4}~(right).
\end{itemize}

\begin{figure}[htb]
\centerline{\epsfxsize=5.5in \epsffile{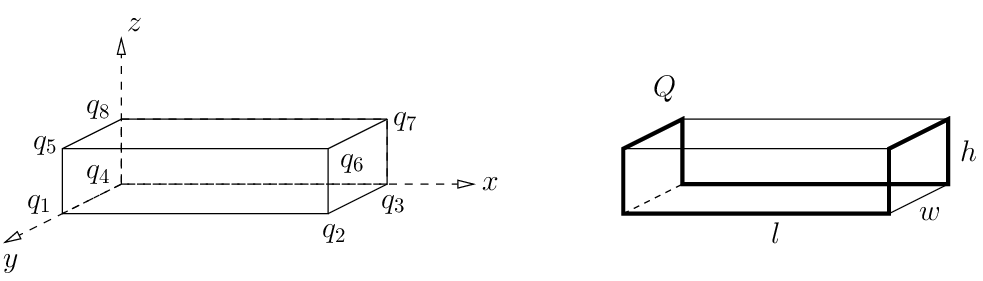}}
\caption{Left: An axis-aligned rectangular box $Q$.
Right: a Hamiltonian cycle (in bold lines) of length $2l +2w +4h$
of the vertices of $Q$ that visits all planes intersecting $Q$.}
\label{f4}
\end{figure}

For each iteration $i=1,\ldots,m$, we compute the box $Q_i$
by linear programming. By a suitable rotation of the set $\H$ of hyperplanes,
the box $Q_i$ is axis-aligned. This can be obtained in $O(n)$ time per iteration.
For a hyperplane $\sigma$, let $\vec{u}(\sigma)$ denote the unit vector
orthogonal to $\sigma$ with a positive $x_d$-coordinate.
An axis-aligned rectangular box in $\RR^d$ has $2^{d-1}$ antipodal
pairs of vertices, which we denote by $s_j$ and $t_j$, for
$j=1,\ldots, 2^{d-1}$, such that the vector $s_jt_j$ has a positive
$x_d$-coordinate. Partition $\H$ into $2^{d-1}$ \emph{types} based on
the following rule (ties are broken arbitrarily):

\begin{itemize} \itemsep 0pt
\item $\sigma\in \H$ is of type $j$, $j\in \{1,\ldots, 2^{d-1}\}$,
      if the $\vec{u}(\sigma)$-minimal and $\vec{u}(\sigma)$-maximal
      vertices of $Q_i$ are $s_j$ and $t_j$, respectively.
\end{itemize}

Let $\H=\bigcup_{j=1}^{2^{d-1}}\H_i$ be the corresponding partition of
the hyperplanes given by this rule. For a hyperplane $\sigma$, that is
not parallel to any coordinate axis, denote by $\sigma(p) \leq 0$
(respectively, by $\sigma(p)\geq 0$) that a point $p\in \RR^d$ lies in the
closed halfspace bounded from above by $\sigma$ (resp., bounded from
below by $\sigma$). Observe that for $j=1,\ldots 2^{d-1}$,
\begin{itemize} \itemsep 0pt
\item a hyperplane $\sigma \in \H_j$ intersects the rectangular box
  $Q_i$ if and only if $\sigma(s_j)\leq 0\leq \sigma(t_j)$.
\end{itemize}

The minimum-perimeter objective is naturally expressed as a linear function.
The resulting linear program has $2d$ variables $x_1,y_1,\ldots, x_d,y_d$
for the box $Q_i =[x_1,y_1] \times \ldots \times [x_d,y_d]$,
and $2n+d$ constraints.

\begin{align*} \label{LP1}
\textup{minimize} \quad  &\sum_{k=1}^d (y_k-x_k)
\ \ \ \quad \quad \textup{(LP1)} \\
\textup{subject to} \quad & \left\{
\begin{array}{lll}
\sigma(s_j)\leq 0 & \mbox{ \rm if }\sigma\in \H_j, &\forall \sigma\in \H\\
\sigma(t_j)\geq 0 & \mbox{ \rm if }\sigma\in \H_j, &\forall \sigma\in \H\\
x_k \leq y_k &&\forall k\in \{1,\ldots, d\}
\end{array}
\right.
\end{align*}

\paragraph{Algorithm analysis.} The key observation is the following.
\begin{observation}\label{O1}
\begin{itemize}\itemsep 0pt
\item[]
\item[{\rm (i)}] If a polygon $\gamma$ intersects $\H$, then
  $\conv(\gamma)$, and any other set containing $\conv(\gamma)$, also intersects~$\H$.
\item[{\rm (ii)}] If a convex polytope $Q$ intersects $\H$, then every
  Hamiltonian cycle of the vertices of $Q$ also intersects $\H$.
\end{itemize}
\end{observation}

Let $Q^*$ be a minimum-perimeter rectangular box intersecting $\H$,
with side lengths denoted by $w_1,\ldots , w_d$.
To account for the error made by discretization, we need the following
easy fact. The planar variant was shown in~\cite[Lemma 2]{DJ10}.
We include the almost identical proof for completeness.

\begin{lemma} \label{lem:box}
There exists $i \in \{1,\ldots,m\}$ such that
$\per(Q_i) \leq (1+\eps)\,\per(Q^*)$.
\end{lemma}
\begin{proof}
Consider a box $Q_i$, $i \in \{1,\ldots, m\}$, that minimizes
the angle difference $\beta$ between the orientations of $Q_i$ and $Q^*$.
By construction, there exists $i \in \{1,\ldots , m\}$ such that
the angle $\beta$ between the orientations of $Q_i$ and $Q^*$
is at most $\eps/(d-1)$, that is, $\beta \leq \eps/(d-1)$.

Let $Q_i'$ be the minimum-perimeter box with the same orientation
as $Q_i$ such that $Q_i'$ contains $Q^*$. By definition,
$\per(Q_i) \le \per(Q'_i)$. An easy trigonometric calculation shows
that the corresponding sides $w_1',\ldots ,w_d'$ of $Q_i'$ are
bounded from above as follows. For $j=1,\ldots , d$, we have
$$ w_j' \leq w_j \cos \beta + \left( \sum_{k\neq j}w_k \right) \sin \beta \leq
w_j + \left( \sum_{k\neq j} w_k \right) \, \frac{\eps}{d-1}.$$
Consequently,
$$ \sum_{j=1}^d w_j' \leq (1+\eps) \sum_{j=1}^d w_j,$$
that is,
$$ \per(Q'_i) \le (1 + \eps)\, \per(Q^*).
$$
Since $\per(Q_i) \le \per(Q'_i)$, it follows that
$\per(Q_i) \le (1 + \eps)\, \per(Q^*)$,
as required.
\end{proof}

\begin{lemma} \label{lem:curve}
A closed curve $\gamma\subset \RR^d$ is contained in a
rectangular box $Q$ with side lengths $w_1,\ldots , w_d$ satisfying
$\sum_{j=1}^d w_j \leq \frac{\sqrt{d}}{2} \, \len(\gamma)$.
Consequently, $\per(Q) \leq \sqrt{d}\cdot 2^{d-2} \, \len(\gamma)$.
\end{lemma}
\begin{proof}
Let $\gamma$ be a closed curve and let $Q=Q(\gamma)$ be a
minimum-perimeter enclosing rectangular box.
Assume for convenience that $Q$ is axis-aligned, so that
its extents in the $d$ coordinates are $w_1,\ldots,w_d$, respectively.
Since $Q$ has minimum perimeter, $\gamma$ meets each $(d-1)$-dimensional
face of $Q$. Arbitrarily select a point $a_i$ of $\gamma$ on each of the
$2d$ faces of $Q$, in the order traversed by $\gamma$, to obtain a polygonal
closed curve $\gamma_1 = (a_1,\ldots,a_{2d})$ still enclosed in $Q$
(duplicate points are possible). For convenience, introduce $a_{2d+1}=a_1$.

By the triangle inequality,
\begin{equation}\label{eq:tri+}
\len(\gamma)\geq \len(\gamma_1) =
\sum_{i=1}^{2d} \len(a_ia_{i+1}).
\end{equation}
By the Cauchy-Schwarz inequality, for $i=1,\ldots ,2d$, we have
\begin{equation}\label{eq:CS+}
\len(a_ia_{i+1})
=\left(\sum_{j=1}^d \Delta_j^2(a_ia_{i+1})\right)^{1/2}
\geq \frac{1}{\sqrt{d}} \sum_{j=1}^d \Delta_j(a_ia_{i+1}).
\end{equation}

Since $\gamma_1$ is a closed curve that visits both faces of $Q$
orthogonal to the $j$th axis for each $j=1,\ldots,d$, we have
$$ \sum_{i=1}^{2d} \Delta_j(a_ia_{i+1}) \geq 2w_j, \text{  for } j=1,\ldots,d. $$
Combined with \eqref{eq:tri+} and \eqref{eq:CS+}, this yields
$ \len(\gamma) \geq \frac{2}{\sqrt{d}} \sum_{j=1}^d w_j$, as claimed.
\end{proof}

Let $L^*=\len(\opt)$ and let $Q_\opt$ be a minimum-perimeter
rectangular box containing $\opt$. By Observation~\ref{O1} and
Lemmas~\ref{lem:box} and~\ref{lem:curve}, we have
\begin{equation}\label{eq:O1L12}
\per(Q_i)\leq (1+\eps)\per(Q^*)\leq (1+\eps)\per(Q_\opt)\leq
(1+\eps)\sqrt{d}\cdot 2^{d-2} L^*.
\end{equation}
By Observation~\ref{O1}, any Hamiltonian cycle of $Q_i$ is a valid tour
of the hyperplanes in $\H$, and its length is bounded above by
$\per(Q_i)$. From \eqref{eq:O1L12}, this length is at most
$(1+\eps)\sqrt{d}\cdot 2^{d-2}$ times the optimum.

We now refine the analysis and show that the length $\tau(Q_i)$ of a \emph{shortest}
Hamiltonian cycle of $Q_i$ is at most $2^{d-1}/\sqrt{d}$ times the optimum.
Algorithm A1 computes a tour $T$ of length $L=\tau(Q_i)=w_d+\sum_{j=1}^d 2^{d-j} w_j$,
where $\sum_{j=1}^d w_j \leq (1+\eps) \frac{\sqrt{d}}{2} L^*$.
For $i=1,\ldots,d$ put $S_i=\sum_{j=1}^i w_j$ and $S=S_d$.
Since $w_1\leq w_2\ldots \leq w_d$, we have $S_i \leq i S/d$, for $i=1,\ldots,d$.
Consequently,
\begin{align} \label{E37}
L &= w_d+\sum_{j=1}^d 2^{d-j} w_j = 2^{d-1} w_1 +
2^{d-2} w_2 + \ldots + 2 w_{d-1} +2w_d \nonumber \\
&= 2S_d + \sum_{i=1}^{d-2} 2^i S_{d-i-1}
\leq \frac{S}{d} \left( 2d + \sum_{i=1}^{d-2} 2^i (d-i-1) \right) \nonumber \\
&=\frac{S}{d} \left( \left(2d + d\sum_{i=1}^{d-2} 2^i \right)
-\sum_{i=1}^{d-2} (i+1) 2^i \right) \nonumber \\
&=\frac{S}{d} \left( d \, 2^{d-1} - (d-2) \, 2^{d-1} \right) = \frac{2^d}{d} \, S.
\end{align}
To evaluate $\sum_{i=1}^{d-2} (i+1) 2^i$ in the last line
of~\eqref{E37}, we set $F(x) = \sum_{i=2}^{d-1} x^i$, and evaluate
its derivative $F'(x)$ in two ways (we omit the details).
Substituting now the upper bound $S \leq (1+\eps) \frac{\sqrt{d}}{2} L^*$
yields
$$ L \leq \frac{2^d}{d} \, S \leq (1+\eps) \frac{\sqrt{d}}{2} \, \frac{2^d}{d} L^*
= (1+\eps) \frac{2^{d-1}}{\sqrt{d}} L^*, $$
as required.

A rough upper estimate on the running time accounts for $m=O(d^d \eps^{1-d})$
$2d$-dimensional linear programs,
each solved in $O(d^2 2^{2d} \, n)$ time~\cite{CM96,MSW96}. The overall
running time is $O(C_{d,\eps}\ n)$, where $C_{d,\eps} = d^2 2^{2d} \, (d/\eps)^d$.

In particular, for $d=3$ and $\eps \leq 0.00022$, we have  $L\leq 2.31 L^*$,
thus algorithm A1 computes a tour  whose length is at most $2.31$ times the optimal.
The algorithm solves a (large!) constant number of $6$-dimensional
linear programs, each in $O(n)$ time~\cite{Me84}. The overall time is $O(n)$.
A modest number of linear programs suffices to get a weaker approximation,
say $2.5$ or $3$.

\paragraph{Remark.} A standard reduction from the \emph{sorting problem}
or from the \emph{convex hull problem} as in~\cite{PS85},
applied to a suitable set of hyperplanes, shows that a shortest TSP
tour for $n$ hyperplanes in $\RR^d$, $d\geq 2$, cannot be computed in $O(n)$ time;
that is, in the worst-case, finding an optimal tour requires $\Omega(n \log{n})$
time in the algebraic decision tree model of computation.

\section{TSPN for lines in $\RR^d$}  \label{sec:lines}

In this section we prove Theorem~\ref{thm:lines}. Let
$\L=\{\ell_1,\ldots , \ell_n\}$ be a set of $n$ lines in $\RR^d$, $d\geq 3$.
If all lines are parallel, we reduce TSPN for $\L$ to TSP for the $n$ intersection
points of the lines with an arbitrary orthogonal hyperplane. Otherwise,
we reduce the TSPN problem to a group Steiner tree problem on a geometric graph.
Specifically, we construct a geometric graph $G_{\L}=(V,E)$, where $V$ is a set of
points on the lines in $\L$, and $E$ consists of line segments connecting some of
these points; the weight of an edge is its Euclidean length.
We have $V = \bigcup_{i=1}^n V_i$, where $V_i \subset \ell_i$
($i=1,\ldots,n$) naturally form $n$ groups, one for each line.
We then run an approximation algorithm for the group Steiner tree
problem on this graph.

It is well known that an optimal TSP tour for points can be 2-approximated
by a minimum spanning tree (a TSP tour is obtained by doubling the edges of
the MST and by using shortcuts and the triangle inequality).
Reich and Widmayer~\cite{RW90} introduced the following
\emph{group Steiner tree} (a.k.a., \emph{one-of-a-set Steiner tree}) problem.
Given an edge weighted graph $G=(V,E)$ and $g$ \emph{groups}
of vertices $V_1,\ldots,V_g \subseteq V$, $|V|=n$, find a tree of
minimum weight in $G$ that includes at least one vertex from each group.
The problem is known to be APX-hard~\cite{BP89}, and it cannot be
approximated better than $\Omega (\log^{2-\eps} n)$ for any $\eps>0$
unless NP admits quasipolynomial-time Las Vegas~algorithms~\cite{HK03}.
The current best approximation ratio, $O(\log^2 n \log g)$ comes from
the algorithm of Garg~\etal~\cite{GKR00} as further refined by
Fakcharoenphol~\etal~\cite{FKR04}. As before with the MST, by doubling
the edges of such a tree, and by using shortcuts and the triangle inequality, one
can obtain a Hamiltonian cycle which includes at least one vertex from
each group, and the approximation ratio of this cycle is of the
same asymptotic order as the approximation ratio of the group Steiner tree used.

The key Lemma~\ref{lem:line} below shows that the length of a minimum group Steiner tree
in $G_{\L}$ (the graph used by the algorithm) is a constant-factor approximation
for the minimum TSP tour for $\L$.
In our case, the graph $G_{\L}$ has $O(n^3)$ vertices and the number of
groups is $n$, so the $O(\log^2 n \log g)$-approximation~\cite{FKR04,GKR00}
for the group Steiner tree problem on a graph with $n$ vertices and $g$ groups
yields an $O(\log^3n)$-approximation for TSPN for $n$ lines in $\RR^d$.

\paragraph{Construction of graph $G_{\L}$.}
A \emph{transversal} between two lines, $\ell_i$ and $\ell_j$, is a
line segment $t_{i,j}t_{j,i}$ with $t_{i,j}\in \ell_i$ and $t_{j,i}\in \ell_j$.
A \emph{minimum transversal} of two lines is one of minimum length; it is orthogonal
to both lines, and if the two lines intersect, it is a segment of zero length
(\ie, $t_{i,j}=t_{j,i}$). A pair of skew lines admits a unique minimal transversal.

We define $G_{\L}$ in terms of a set $S$ of transversal segments
among the lines: let the vertices of $G_{\L}$ be the set of
endpoints of the segments in $S$; the edges of $G_{\L}$ include all
segments in $S$, and all segments along the lines in $\L$ between
consecutive vertices. We use two types of transversal segments, $S_1$ and $S_2$,
with $S=S_1\cup S_2$. Let $S_1$ be the set of minimum transversals
between all pairs of nonparallel lines in $\L$. To define
$S_2$, we proceed as follows; see Fig.~\ref{fig:lines-block}~(left).
For each ordered pair of nonparallel lines $(\ell_i,\ell_j)\in \L^2$,
let $T_{i,j}$ be the hyperplane orthogonal to $\ell_i$ and containing
$t_{i,j}$, and let $P_{i,j}=\{T_{i,j}\cap \ell: \ell\in \L\}$ be the
set of intersection points of $T_{i,j}$ with the lines in $\L$.
Note that $t_{i,j}\in P_{i,j}$ and $|P_{i,j}|\leq n$.
Add all edges of the complete graph on $P_{i,j}$ to the set $S_2$.
\begin{figure}[htb]
\centerline{\epsfxsize=.99\textwidth \epsffile{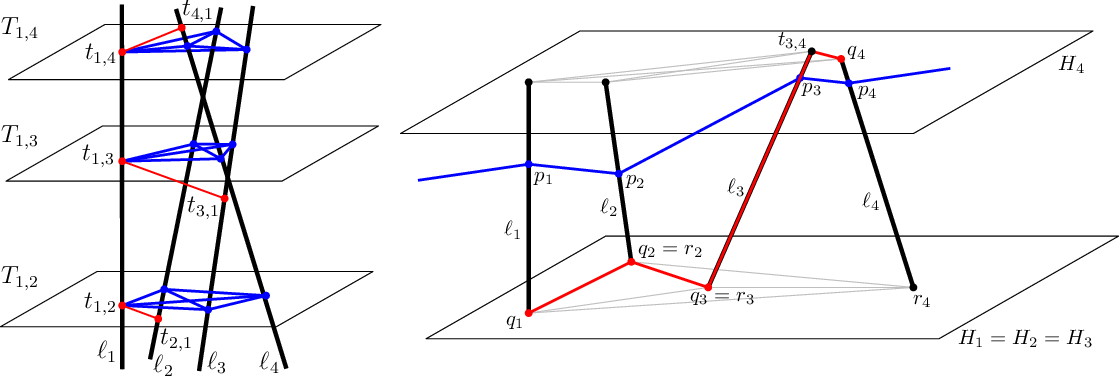}}
\caption{Left: A set of four lines $\ell_1,\ldots ,\ell_4$.
The minimum transversals between $\ell_1$ and the
other three lines are $t_{1,2}t_{2,1}$, $t_{1,3}t_{3,1}$ and $t_{1,4}t_{4,1}$.
For $i=1,2,3$, we insert a complete graph in the
plane $T_{1,i}$ orthogonal to $\ell_1$ and incident to $t_{1,i}$.
Right:  An optimal tour $\opt$ visits the lines $\ell_1,\ell_2,\ell_3,\ell_4$
 at points $p_1,p_2,p_3,p_4$, respectively. We construct a path
 $\gamma_1=(q_1 q_2 q_3 q_4)$ that visits these lines in the same
 order. Points $q_1,q_2,q_3$ are in the same hyperplane
 $H_1=H_2=H_3$. Point $q_4$ is in a hyperplane $H_4\neq H_3$
 because $|q_3r_4| >3 |p_3p_4|$.}
\label{fig:lines-block}
\end{figure}

The set of transversals $S=S_1\cup S_2$ determines $G_{\L}$.
The segments in $S_1$ have at most $n(n-1)$ endpoints, and for
each segment endpoint we compute a complete graph, each with at most $n$ vertices
and ${n\choose 2}$ edges, thus $S_2$ contains $O(n^4)$ segments
and $|S|= |S_1 \cup S_2|= O(n^2 + n^4) = O(n^4)$.
Consequently, $G_{\L}$ has $O(n^3)$ vertices and $|S|+ O(n^3) = O(n^4)$ edges.
The vertices in $G_{\L}$ are partitioned into $n$ groups, one for each
line $\ell \in \L$. The group corresponding to line $\ell$ contains
all $O(n^2)$ endpoints of transversal segments in $S$ on $\ell$.

The minimum transversal of two skew lines in $\RR^d$ lies in the 3-dimensional
affine subspace spanned by the lines, and it can be computed in $O(d)$ time;
point-hyperplane intersections can also be computed in $O(d)$ time in $\RR^d$.
Consequently, the graph $G_{\L}$ can we computed in $O(d n^4)$ time.

\paragraph{Two technical lemmas.} The approximation relies on Lemmas~\ref{lem:mt}
and~\ref{lem:par} (below). According to Lemma~\ref{lem:mt}, if the
directions of two lines are far apart, then a connecting segment can
be approximated by a $3$-segment path that detours through the minimum transversal
of the two lines. According to Lemma~\ref{lem:par}, if two lines are
nearly vertical and we are given some horizontal transversal segment
between the lines, then the only way to find a significantly shorter transversal
is to move the endpoints closer to the endpoints of the minimum transversal.

\begin{figure}[htb]
\centerline{\epsfxsize=6in \epsffile{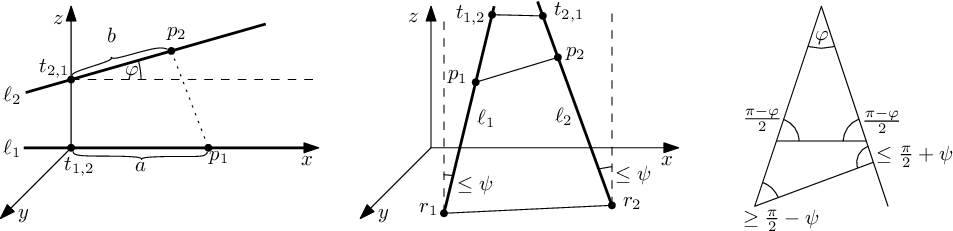}}
\caption{Left: The angle between lines $\ell_1$ and $\ell_2$ is $\varphi$.
The distance between $p_1\in \ell_1$ and $p_2\in \ell_2$ is approximated
by the polygonal path $(p_1,s_1,s_2,p_2)$ that passes through the
minimum transversal $s_1s_2$ between the two lines.
Middle: If $|p_1p_2|\leq \frac{1}{3}|r_1r_2|$, then
$p_1$ and $p_2$ are much closer to the minimum transversal
than $r_1$ and $r_2$, respectively.
Right: Two triangles with angle $\varphi$. The other two angles are equal
in one triangle, and they differ by at most $2\psi$ in the other triangle.}
\label{fig:mt}
\end{figure}

\begin{lemma}\label{lem:mt}
Let $\ell_1$ and $\ell_2$ be two lines in $\RR^d$ such that the angle between their
directions is $\varphi \in (\varphi_0,\frac{\pi}{2}]$.
  Let $t_{1,2}t_{2,1}$ be their minimum transversal with
  $t_{1,2}\in \ell_1$ and $t_{2,1}\in \ell_2$.
Let $p_1\in \ell_1$ and $p_2\in \ell_2$ be two points. Then
$|p_1t_{1,2}|+|t_{1,2}t_{2,1}|+|t_{2,1}p_2|\leq \sqrt{\frac{3}{1-\cos \varphi_0}} |p_1p_2|.$
\end{lemma}
\begin{proof}
Consider the 3-dimensional affine subspace spanned by $\ell_1$ and $\ell_2$.
Without loss of generality, we may assume that $\ell_1$ is the $x$-axis,
$p_1=(a,0,0)$, $t_{1,2}=(0,0,0)$, and $t_{2,1}=(0,0,h)$ as in Fig.~\ref{fig:mt}~(left).
Let $a=|p_1 t_{1,2}|$ and $b=|p_2 t_{2,1}|$.
The Cauchy-Schwarz inequality yields the upper bound
$$ |p_1t_{1,2}|+|t_{1,2}t_{2,1}|+|t_{2,1}p_1|= a+h+b \leq  \sqrt{3(a^2+b^2+h^2)}. $$
If $x(p_2) \geq 0$ (as in Fig.~\ref{fig:mt}, left), then the law of cosines yields
\begin{eqnarray}
|p_1p_2|^2 &=& h^2+a^2+b^2-2ab\cos \varphi\nonumber\\
&=& h^2+ (a-b)^2\cos \varphi+ (a^2+b^2)(1-\cos \varphi)\nonumber\\
&\geq & (1-\cos\varphi) (a^2+b^2+h^2)\nonumber\\
&\geq & (1-\cos\varphi_0) (a^2+b^2+h^2).\nonumber
\end{eqnarray}
If $x(p_2) \leq 0$, then
$|p_1p_2|^2 = h^2+a^2+b^2-2ab\cos (\pi-\varphi) \geq h^2+a^2+b^2$,
since $\cos(\pi-\varphi)<0$, and we obtain
$|p_1p_2|^2 \geq (1-\cos\varphi_0) (a^2+b^2+h^2)$
in this case, as well. In both cases, the claimed inequality follows
after taking square roots.
\end{proof}

\begin{lemma}\label{lem:par}
Let $\ell_1$ and $\ell_2$ be two lines in $\RR^d$ such that the angle between their
directions is $\varphi\in (0,\frac{\pi}{6}]$; and the direction of each line
differs from the $x_d$-axis by at most $\psi\in [0,\frac{\pi}{6}]$.
Let $p_1\in \ell_1$ and $p_2\in \ell_2$ be two arbitrary points on the two lines;
let $r_1\in \ell_1$ and $r_2\in \ell_2$ be the intersection points of the two lines
with a hyperplane orthogonal to the $x_d$-axis; and $t_{1,2}t_{2,1}$ be
the minimum transversal
of the two lines such that $t_{1,2}\in \ell_1$ and $t_{2,1}\in \ell_2$
(Fig.~\ref{fig:mt}, middle). If $3|p_1p_2|\leq |r_1r_2|$, then
$|p_1t_{1,2}|\leq \frac{2\sqrt{3}}{9}\ |r_1t_{1,2}|$
and $|p_2t_{2,1}|\leq \frac{2\sqrt{3}}{9}\ |r_2t_{2,1}|$.
\end{lemma}
\begin{proof}
Let $h=|t_{1,2}t_{2,1}|$ be the distance between the two lines.
Put $a=|p_1t_{1,2}|$, $b=|p_2t_{2,1}|$, $e=|r_1t_{1,2}|$, and $f=|r_2t_{2,1}|$.
Let $\varphi_p\in \{\varphi, \pi-\varphi\}$ be the angle between the rays
$\overrightarrow{t_{1,2}p_1}$ and $\overrightarrow{t_{2,1}p_2}$.
By the law of cosines, we have $|p_1p_2|^2=h^2+a^2+b^2-2ab\cos \varphi_p$.
The sum of the last three terms in this expression is
$c^2=a^2+b^2-2ab\cos \varphi_p$, where $c$ is the
third side of a triangle with two adjacent sides of lengths $a$ and $b$
that meet at angle $\varphi_p$. Denote by $\beta$ the angle of this triangle
opposite to the \emph{longer} of $a$ and $b$. Then the law of sines yields
$$ a^2+b^2-2ab\cos\varphi_p =
(\max\{a,b\})^2\cdot \frac{\sin^2\varphi_p}{\sin^2\beta}
\geq (\max\{a,b\})^2 \sin^2\varphi_p
= (\max\{a,b\})^2 \sin^2\varphi. $$
Consequently,
\begin{equation} \label{E23}
|p_1p_2|^2 \geq h^2 + (\max\{a,b\})^2\sin^2 \varphi.
\end{equation}

Let $\varphi_r\in \{\varphi, \pi-\varphi\}$ be the angle between
the rays $\overrightarrow{t_{1,2}r_1}$ and $\overrightarrow{t_{2,1}r_2}$.
We show that $\varphi_r=\varphi$. Indeed, since $r_1r_2$ lies in a
hyperplane orthogonal to the $x_d$-axis, and the direction of each line
differs from the $x_d$-axis by at most $\psi\in [0,\frac{\pi}{6}]$,
the directions of $r_1r_2$ and the minimal transversal $t_{1,2}t_{2,1}$
differ by at most $\frac{\pi}{6}$. If $\varphi_r=\pi-\varphi\geq \frac{5\pi}{6}$,
then $|r_1t_{1,2}|\leq h\tan \frac{\pi}{6}$ and $|r_2t_{2,1}|\leq h\tan \frac{\pi}{6}$.
The triangle inequality yields
$|r_1r_2|\leq |r_1t_{1,2}|+|t_{1,2}t_{2,1}|+|t_{2,1}r_2|
\leq (1+2\tan\frac{\pi}{6})h= (1+\frac{2\sqrt{3}}{3})h\leq 2.16h$,
in contradiction with the assumed inequality $|r_1r_2|\geq 3|p_1p_2|\geq 3h$.

By the law of cosines we have
$$ |r_1r_2|^2=h^2+e^2+f^2-2ef\cos\varphi_r =h^2+e^2+f^2-2ef\cos \varphi. $$
Consider a triangle where $e$ and $f$ are adjacent sides
parallel with $r_1t_{1,2}$ and $r_2t_{2,1}$ respectively,
that meet at angle $\varphi$. Since the directions of $\ell_1$ and $\ell_2$
differ from vertical by at most $\psi$, the angle opposite to the
shorter of $e$ and $f$ is at least $\frac{\pi}{2}-\psi$ (see Fig.~\ref{fig:mt}, right).
Hence the law of sines yields  $e^2+f^2-2ef\cos \varphi \leq
(\min\{e,f\})^2\cdot \frac{\sin^2\varphi}{\sin^2(\pi/2-\psi)}$, and
consequently
\begin{equation} \label{E24}
|r_1r_2|^2 \leq
h^2 + (\min\{e,f\})^2\cdot \frac{\sin^2\varphi}{\sin^2 (\pi/2-\psi)}.
\end{equation}

The inequality $3|p_1p_2|\leq |r_1r_2|$ in combination with
inequalities~\eqref{E23} and~\eqref{E24} implies
\begin{eqnarray}
9\left(h^2 + \left(\max\{a,b\}\right)^2 \sin^2 \varphi \right)
&\leq& h^2 +\left(\min\{e,f\}\right)^2\frac{\sin^2\varphi}{\sin^2 (\pi/2-\psi)}, \nonumber\\
9\left(\max\{a,b\}\right)^2 \sin^2 \varphi
&\leq& \left(\min\{e,f\}\right)^2\frac{\sin^2\varphi}{\sin^2 (\pi/2-\psi)},
\end{eqnarray}

and further (after canceling $\sin^2\varphi$ and taking square roots) that
\begin{equation} \label{E25}
\max\{a,b\} \leq \frac{\min\{e,f\}}{3\sin(\pi/2-\psi)}.
\end{equation}

If $\psi\in [0,\frac{\pi}{6}]$, then
$\sin\left(\frac{\pi}{2}-\psi\right) \geq \sin\left(\frac{\pi}{3}\right)
  =\frac{\sqrt{3}}{2}$
and so $\max\{a,b\}\leq \frac{2\sqrt{3}}{9}\min\{e,f\}$.
It follows that $a\leq \frac{2\sqrt{3}}{9}e$ and $b\leq \frac{2\sqrt{3}}{9}f$,
as required.
\end{proof}

\paragraph{Group Steiner tree yields a constant-factor approximation
  for TSP with lines.} The main result of this section is the following lemma.
Theorem~\ref{thm:lines} then directly follows from this lemma.
\begin{lemma}\label{lem:line}
Let $\L$ be a set of $n$ lines in $\RR^d$. Then the length of a
minimum group Steiner tree in $G_{\L}$ is $O(1)$ times the length
of a minimum TSP tour for the lines in $\L$.
\end{lemma}
\begin{proof}
Let $\L=\{\ell_1,\ldots , \ell_n\}$, where the lines are indexed so that
an optimal TSP tour is $\opt(\L)=(p_1,\ldots , p_n)$ with $p_i\in \ell_i$,
$i=1,\ldots , n$. We show that $G_{\L}$ contains a group Steiner tree $T$
of length at most $83 \, \len(\opt(\L))$.
The argument does not use the optimality of the tour $\opt(\L)$,
\ie, for any cycle $C=(p_1,\ldots , p_n)$, $p_i\in \ell_i$, we construct
a group Steiner tree $T$ of length at most $83 \, \len(C)$.
The tree $T$ consists of a main (backbone) path $\gamma_0$, and a path attached
to each vertex of the backbone.

Decompose the cyclic sequence $(\ell_1,\ldots , \ell_n)$
into maximal subsequences, called \emph{blocks},
$$ (\ell_{\tau(i)}, \ell_{\tau(i)+1},\ldots ,
\ell_{\tau(i+1)-1}), \hspace{0.5in} i=1,2,\ldots ,k $$
for some $k\geq 1$ as follows. Let $\tau(1)=1$, and for each
$i=1,\ldots ,k-1$, let $\tau(i+1)$ be the first
index such that the directions of $\ell_{\tau(i)}$ and
$\ell_{\tau(i+1)}$ differ by more than $\frac{\pi}{12}$.
That is, the directions of the lines in the $i$-th block differ
from the direction of $\ell_{\tau(i)}$ by at most $\frac{\pi}{12}$.
By the triangle inequality, the directions of any two lines
in a block differ by at most $\frac{\pi}{6}$ (as required by Lemma~\ref{lem:par}).

Consider the sequence of the first elements of the blocks,
$(\ell_{\tau(1)}, \ell_{\tau(2)}, \ldots , \ell_{\tau(k)})$.
By construction, the directions of any two consecutive lines in the above sequence
differ by more than $\frac{\pi}{12}$. If $k\geq 2$, the ``backbone'' of the
group Steiner tree $T$ is the polygonal path
$$\gamma_0=
(t_{\tau(1),\tau(2)} \, t_{\tau(2),\tau(1)} \, t_{\tau(2),\tau(3)}
\,  t_{\tau(3),\tau(2)} \ldots  t_{\tau(k-1),\tau(k)}\, t_{\tau(k),\tau(k-1)}).$$
Lemma~\ref{lem:mt} with $\varphi_0=\frac{\pi}{12}$ implies that
$\len(\gamma_0)$ is bounded from above as follows:
\begin{eqnarray}
  \len(\gamma_0)&=& \len(t_{\tau(1),\tau(2)}\, t_{\tau(2),\tau(1)}\, t_{\tau(2),\tau(3)}\, t_{\tau(3),\tau(2)}
  \ldots  t_{\tau(k-1),\tau(k)}\, t_{\tau(k),\tau(k-1)})\nonumber\\
&\leq & \len(p_{\tau(1)}\, t_{\tau(1),\tau(2)}\, t_{\tau(2),\tau(1)}\, p_{\tau(2)})+  \ldots +
                     \len(p_{\tau(k-1)}\, t_{\tau(k-1),\tau(k)}\, t_{\tau(k),\tau(k-1)}\, p_{\tau(k)})\nonumber\\
&\leq& \sqrt{\frac{3}{1-\cos(\pi/12)}} \ \len(p_{\tau(1)} p_{\tau(2)} \ldots  p_{\tau(k)}) \nonumber\\
&\leq& 9.4\ \len(p_1 p_2 \ldots  p_n) \leq 9.4\ \len(C).\label{E:U}
\end{eqnarray}

For each block
$(\ell_{\tau(i)},\ell_{\tau(i)+1},\ldots , \ell_{\tau(i+1)-1})$, $i=1,\ldots,k$,
we attach a path $\gamma_i$ visiting the lines in this block
to the backbone $\gamma_0$. Each path $\gamma_i$ is constructed incrementally
starting from an initial vertex and an initial hyperplane containing that vertex.
If $k \geq 2$, then $\gamma_i$ starts from vertex $t_{\tau(i),\tau(i+1)}\in
\ell_{\tau(i)}\cap \gamma_0$ within hyperplane $T_{\tau(i),\tau(i+1)}$ for $i=1,\ldots , k-1$;
and $\gamma_k$ starts from vertex $t_{\tau(k),\tau(k-1)}\in \ell_{\tau(k)}\cap \gamma_0$
within hyperplane $T_{\tau(k),\tau(k-1)}$. If $k=1$ (i.e., there is only one block),
then $\gamma_0$ is not needed, and we set $T:=\gamma_1$. In this case, we construct
a path $\gamma_1$ starting from each of the $O(n^2)$ vertices on $\ell_1$ and every
possible hyperplane of the form $T_{i,j}$ containing that vertex; and then show that
one of these paths satisfies $\len(\gamma_1)\leq 9.77\,\len(C)$.

The paths $\gamma_i$, $i=1,\ldots , k$, are constructed analogously apart from the
choice of their initial vertex $q_1$ and initial hyperplane $H_1$, $q_1\in H_1$.
We explain the construction for $i=1$ only.
Consider the first block, $(\ell_1,\ell_2,\ldots,\ell_m)$, where $1\leq m\leq n$.
The path $\gamma_1$ will use transversal segments from $S_2$ between lines in
$\ell_1,\ldots , \ell_m$, and possibly some edges along the lines $\ell_1,\ldots , \ell_m$.

We construct $\gamma_1$ incrementally for a given initial vertex $q_1$ and hyperplane $H_1$,
$q_1\in H_1$. Refer to Fig.~\ref{fig:lines-block}~(right).
In each step, we maintain a vertex $q_i\in \ell_i$ of $\gamma_1$ and hyperplane $H_i$
such that $q_i\in H_i$ and $G_{\L}$ contains a complete graph on the intersection
points between the lines in $\L$ and $H_i$. Initially, we have a single-vertex path
$\gamma_1=(q_1)$, where $q_1\in \ell_1$ and $q_1\in H_1$.
We extend $\gamma_1$ in $m-1$ steps to visit some points $q_i\in \ell_i$, $i=2,\ldots , m$.
In step $i$, we would like to extend $\gamma_1$ from $q_i$ to $q_{i+1}\in \ell_{i+1}$
by a single edge in $H_i$. However, if the distance from $q_i$ to $\ell_{i+1}\cap H_i$
is more than $3\ |p_i p_{i+1}|$, then $\gamma_1$ will follow $\ell_i$ to the endpoint
$t_{i,i+1}$ of the minimal transversal between $\ell_i$ and $\ell_{i+1}$,
and reach $\ell_{i+1}$ in the hyperplane $T_{i,i+1}$ (orthogonal to $\ell_i$).

Assume that we have already built the path $\gamma_1$ up to vertex $q_i\in \ell_i$,
$i\in \{1,\ldots , m-1\}$, with a hyperplane $H_i$, $q_i\in H_i$.
We choose $q_{i+1}\in \ell_{i+1}$, the portion of $\gamma_1$ from $q_i$
to $q_{i+1}$, and the hyperplane $H_{i+1}$ as follows.
Let $r_{i+1}=\ell_{i+1}\cap H_i$ (note that $r_{i+1}$ is
a vertex of $G_{\L}$ by construction). We distinguish two cases:
\begin{itemize} \itemsep 0pt
\item If $|q_i r_{i+1}|\leq 3\ |p_i p_{i+1}|$, then let
  $q_{i+1}=r_{i+1}$, extend the path $\gamma_1$ with
  the edge $q_iq_{i+1}\subset H_i$,
  and let $H_{i+1}=H_i$.
\item Otherwise let $H_{i+1}=T_{i,i+1}$ (the hyperplane orthogonal to $\ell_i$
    and containing $t_{i,i+1}$), and let $q_{i+1}=\ell_{i+1}\cap H_{i+1}$.
  Now extend $\gamma_1$ with the segments $q_it_{i,i+1}\subset \ell_i$ and
  $t_{i,i+1}q_{i+1}\subset H_{i+1}$.
\end{itemize}

For estimating $\len(\gamma_1)$, we consider the transversal
segments and the edges along the lines in $\L$ separately.
The length of the transversal segment between $\ell_i$ and
$\ell_{i+1}$ is at most $3 \, |p_ip_{i+1}|$, and consequently,
the total length of all transversal segments in $\gamma_1$
is at most $3 \, \len(p_1 \ldots p_m)$. Indeed, in the first case $H_i$
contains the edge $q_iq_{i+1}$ of length $|q_iq_{i+1}|=|q_ir_{i+1}|\leq 3\ |p_ip_{i+1}|$.
In the second case, $H_{i+1}$ contains segment $t_{i,i+1}\, q_{i+1}$
of length
$$ |t_{i,i+1}\, q_{i+1}| \leq \frac{1}{\cos(\pi/6)}\ |t_{i,i+1}\, t_{i+1,i}|=
\frac{2}{\sqrt{3}}\ |t_{i,i+1}\, t_{i+1,i}| \leq \frac{2}{\sqrt{3}}\ |p_i p_{i+1}|, $$
where the first inequality
follows from the fact that the directions of $\ell_i$ and $\ell_{i+1}$ differ
by at most $\frac{\pi}{6}$, and so the right triangle $\Delta{t_{i,i+1}t_{i+1,i}q_{i+1}}$
has an interior angle at most $\frac{\pi}{6}$ at $t_{i,i+1}$.

It remains to bound the total length of the edges in
$\gamma_1$ that lie along the lines $\ell_1,\ldots,\ell_m$.
Let $1\leq \sigma(1)<\ldots <\sigma(h)< m$ be the subsequence
of indices such that $q_{\sigma(i)}t_{\sigma(i),\sigma(i)+1}\subset \gamma_1$
for $i=1,\ldots,h$; and put $\sigma(0)=1$ (possibly $\sigma(0)=\sigma(1)$).
By construction, the vertices $q_1,\ldots ,q_{\sigma(1)}$ of $\gamma_1$
lie in the same hyperplane $H_1=\ldots =H_{\sigma(1)}$.
We introduce a shorthand notation for the transversals:
for $i=1,\ldots, h$, let $s_{\sigma(i)}=t_{\sigma(i),\sigma(i)+1}$.
With this notation, the length of the edges in $\gamma_1$
that lie along the lines $\ell_1,\ldots,\ell_m$ is precisely
$$Z_1=\sum_{i=1}^h |q_{\sigma(i)}\, s_{\sigma(i)}|.$$

By construction, $\gamma_1$ contains the segment
$q_{\sigma(i)}s_{\sigma(i)}\subset \ell_{\sigma(i)}$
when $|p_{\sigma(i)}\, p_{\sigma(i)+1}|< \frac{1}{3}|q_{\sigma(i)}\, r_{\sigma(i)+1}|$.
In this case, Lemma~\ref{lem:par} is applicable, and it
gives $|p_{\sigma(i)}\, s_{\sigma(i)}|\leq \frac{2\sqrt{3}}{9} |q_{\sigma(i)}\, s_{\sigma(i)}|$.

For $i=1,\ldots ,m$, let $\proj_i:\RR^d\rightarrow \ell_i$
be the projection onto the line $\ell_i$ along the hyperplane $H_i$.
In particular for $i=1,\ldots ,h$, we have $q_{\sigma(i)}=\proj_{\sigma(i)} s_{\sigma(i-1)}$,
and $q_{\sigma(1)}=\proj_{\sigma(1)} q_1$, where $q_1$ is the first vertex of $\gamma_1$.
Recall that the directions of the lines $\ell_{\sigma(i)}$, $i=1,\ldots ,h$,
differ by at most $\frac{\pi}{6}$ from each other.
Consequently, for any line segment $ab$, we have
$|\proj_{\sigma(i)}(ab)|\leq |ab|/\cos\frac{\pi}{6}=\frac{2\sqrt3}{3}|ab|$.
Obviously, for any line segment $ab\subset \ell_{\sigma(i)}$,
we have $|\proj_{\sigma(i)} (ab)|=|ab|$.

We now bound $Z_1$ from above: intuitively, we estimate $|q_{\sigma(i)}\, s_{\sigma(i)}|$
by making a detour via $p_{\sigma(i-1)}\, p_{\sigma(i)}$, which can be related to
the optimal tour. This leads to an upper bound on $Z_1$ in terms of $\len(p_1\ldots p_m)$.

\begin{eqnarray}
Z_1 & = & \sum_{i=1}^h |q_{\sigma(i)}\, s_{\sigma(i)}|
  =                 |\proj_{\sigma(1)} (q_1\, s_{\sigma(1)})|
      +\sum_{i=2}^h |\proj_{\sigma(i)} (s_{\sigma(i-1)}\, s_{\sigma(i)})| \label{eq:Z}\\
&\leq &          \left(|\proj_{\sigma(1)} (q_1 p_1)|+
                 |\proj_{\sigma(1)} (p_1 p_{\sigma(1)})|+
                 |\proj_{\sigma(1)} (p_{\sigma(1)} s_{\sigma(1)})| \right)+ \nonumber\\
&    & \sum_{i=2}^h \left(|\proj_{\sigma(i)} (s_{\sigma(i-1)}p_{\sigma(i-1)})|+
                   |\proj_{\sigma(i)} (p_{\sigma(i-1)} p_{\sigma(i)})|+
                   |\proj_{\sigma(i)} (p_{\sigma(i)} s_{\sigma(i)})| \right)\nonumber\\
&\leq &\frac{2\sqrt3}{3}\, |q_1 p_1| +
       \frac{2\sqrt3}{3}\,\sum_{i=2}^h |s_{\sigma(i-1)}p_{\sigma(i-1)}|+
       \frac{2\sqrt3}{3}\,\sum_{i=1}^h |p_{\sigma(i-1)} p_{\sigma(i)}|+
       \sum_{i=1}^h  |s_{\sigma(i)} p_{\sigma(i)} | \nonumber\\
&\leq &  \frac{2\sqrt3}{3}\, |q_1 p_1| +
         \frac{2\sqrt3}{3}\, \len(p_{\sigma(0)}p_{\sigma(1)}\ldots p_{\sigma(h)})+
         \left(\frac{2\sqrt3}{3}+1\right)\frac{2\sqrt{3}}{9} \,
         \sum_{i=1}^h |q_{\sigma(i)} s_{\sigma(i)}|\nonumber\\
&\leq &  \frac{2\sqrt3}{3}\, |q_1 p_1| +  \frac{2\sqrt3}{3}\, \len(p_1\ldots p_m)+
         \frac{4+2\sqrt3}{9} \, Z_1,\nonumber
\end{eqnarray}
where we used the triangle inequality. After rearranging, we obtain
\begin{equation}\label{eq:Z3}
Z_1
\leq \frac{6(6+5\sqrt3)}{13}\,  \left( |q_1 p_1| + \len(p_1\ldots p_m)\right)
\leq 6.77\, \left( |q_1 p_1| + \len(p_1\ldots p_m)\right).
\end{equation}
It remains to bound the term $|q_1p_1|$ in \eqref{eq:Z3},
which depends on the choice of the initial vertex $q_1$ of $\gamma_1$.
We distinguish two cases.

\paragraph{Case~1: $k\geq 2$ (there are two or more blocks).}
Since $q_1=t_{1,m+1}$ in the first block, Lemma~\ref{lem:mt} yields
\begin{eqnarray} \label{eq:gamma1}
|q_1p_1| = |t_{1,m+1} p_1| &\leq& \len(p_1 t_{1,m+1}t_{m+1,1}p_{m+1})
\leq 9.4 \, |p_1p_{m+1}| =9.4 \, |p_{\tau(1)}\, p_{\tau(2)}|,  \nonumber \\
  Z_1 &\leq& 6.77 \cdot 10.4 \, \len(p_{\tau(1)}, \ldots, p_{\tau(2)}) \leq
   70.5\,\len(p_{\tau(1)}, \ldots, p_{\tau(2)}), \nonumber \\
   \len(\gamma_1) &\leq& Z_1 + 3 \, \len(p_{\tau(1)}, \ldots, p_{\tau(2)})
   \leq 73.5 \, \len(p_{\tau(1)}, \ldots, p_{\tau(2)}).
\end{eqnarray}

Analogous bounds hold for each of the first $k-1$ blocks. The last block requires
a different argument. The term $|q_1p_1|$ in \eqref{eq:Z3} corresponds
to $|t_{\tau(k-1),\tau(k)}\, p_{\tau(k-1)}|$ and
$|t_{\tau(k),\tau(k-1)}\, p_{\tau(k)}|$, respectively, in the last two blocks.
By Lemma~\ref{lem:mt}, the sum of these two terms is bounded by
\begin{equation}\label{eq:pq2}
\len(p_{\tau(k-1)} t_{\tau(k-1),\tau(k)}\, t_{\tau(k),\tau(k-1)}\, p_{\tau(k)})
\leq 9.4 \, |p_{\tau(k-1)}\, p_{\tau(k)}|.
\end{equation}
Summing over all $k$ blocks, the combination of \eqref{eq:gamma1}
and \eqref{eq:pq2} yields
\begin{equation} \label{E36}
  \sum_{i=1}^k \len(\gamma_i) \leq
  73.5\, \sum_{i=1}^{k-1} \len(p_{\tau(i)}, \ldots, p_{\tau(i+1)})
   \leq 73.5\, \len(C).
\end{equation}
Summing~\eqref{E:U} and~\eqref{E36}, we conclude that $G_{\L}$ contains
a group Steiner tree for $\L$ of length
$$ \len(\gamma_0) + \sum_{i=1}^k \len(\gamma_i) \leq
\left(9.4 + 73.5\right) \len(C) \leq 83 \, \len(C). $$

\paragraph{Case~2. $k=1$ (there is only one block).}
In this case, we have $m=n$. Recall that each vertex $q\in \ell_1$ in $G_{\L}$
is the intersection of line $\ell_1$ and some hyperplane $T_{i,j}$,
$i,j\in \{1,\ldots , n\}$.  For every vertex $q\in \ell_1$ and every hyperplane
$H$ of this form containing $q$, let $\gamma_1=\gamma_1(q,H)$ be the
path produced by the incremental process discussed above.
Note that $\gamma_1(q,H)$ visits $\ell_1,\ldots ,\ell_n$ in this order,
and the total length of transversal segments along
$\gamma_1(q,H)$ is at most $3\len(C)$ by construction.
If there is a vertex $q\in \ell_1$ and a hyperplane $H$
for which $\gamma_1(q,H)$ consists of transversal segments
only, then it is a Stener tree for $\ell_1,\ldots ,\ell_n$ of length
$\len(\gamma_1(q,H))\leq 3\, \len(C)$, as required.
Otherwise, denote by $Z_1(q,H)$ the total length
of the edges of $\gamma_1(q,H)$ along the lines in $\L$.

Extend each path $\gamma_1(q,H)$ from its last vertex in $\ell_n$ to a vertex
$q_{n+1}\in \ell_1$ in a hyperplane $H_{n+1}$ by performing one more iteration.
Denote by $\widehat{\gamma}_1(q,H)$ the resulting path.
Suppose that there is a vertex $q\in \ell_1$ and a hyperplane $H$
such that $q_{n+1}=q$ and $H_{n+1}=H$. In this case, $\widehat{\gamma}_1(q,H)$
is a tour. Since the vertices $s_{\sigma(h)}$, $q_{n+1}=q$, and $q_{\sigma(1)}$
are in the hyperplane $H_{n+1}=H=H_{\sigma(1)}$, we have
$\proj_{\sigma(1)} s_{\sigma(h)}=\proj_{\sigma(1)} q_{\sigma(1)}$;
and \eqref{eq:Z} can be replaced by
\begin{eqnarray}
Z_1(q,H) & = & \sum_{i=1}^h |q_{\sigma(i)}\, s_{\sigma(i)}|
      =                |\proj_{\sigma(1)} (s_{\sigma(h)}\, s_{\sigma(1)})|
      +\sum_{i=2}^h |\proj_{\sigma(i)} (s_{\sigma(i-1)}\, s_{\sigma(i)})| \label{eq:Z4}\\
&\leq &          \left(|\proj_{\sigma(1)} (s_{\sigma(h)}\, p_{\sigma(h)})|+
                 |\proj_{\sigma(1)} (p_{\sigma(h)}\, p_{\sigma(1)})|+
                 |\proj_{\sigma(1)} (p_{\sigma(1)}\, s_{\sigma(1)})| \right)+ \nonumber\\
&    & \sum_{i=2}^h \left(|\proj_{\sigma(i)} (s_{\sigma(i-1)}p_{\sigma(i-1)})|+
                   |\proj_{\sigma(i)} (p_{\sigma(i-1)} p_{\sigma(i)})|+
                   |\proj_{\sigma(i)} (p_{\sigma(i)} s_{\sigma(i)})| \right)\nonumber\\
&\leq & \frac{2\sqrt3}{3}\,\sum_{i=2}^{h+1} |s_{\sigma(i-1)}p_{\sigma(i-1)}|+
       \frac{2\sqrt3}{3} \left( |p_{\sigma(h)} p_{\sigma(1)}|+
       \sum_{i=2}^h |p_{\sigma(i-1)} p_{\sigma(i)}|\right) +
       \sum_{i=1}^h  |s_{\sigma(i)} p_{\sigma(i)} | \nonumber\\
&\leq &  \frac{2\sqrt3}{3}\, \len(C)+
         \left(\frac{2\sqrt3}{3}+1\right)\frac{2\sqrt{3}}{9} \,
         \sum_{i=1}^h |q_{\sigma(i)} s_{\sigma(i)}|\nonumber\\
&\leq &   \frac{2\sqrt3}{3}\, \len(C)+
         \frac{4+2\sqrt3}{9} \, Z_1(q,H).\label{eq:Z5}
\end{eqnarray}
After rearranging, we obtain
\begin{equation}\label{eq:Z6}
Z_1(q,H)
\leq \frac{6(6+5\sqrt3)}{13}\, \len(C)
\leq 6.77\, \len(C).
\end{equation}

Even if $\widehat{\gamma}_1(q,H)$ is not a tour for any vertex $q\in \ell_1$
and hyperplane $H$, the concatenation of some paths $\widehat{\gamma}_1(q,H)$
forms a cycle that we denote by $\Gamma$. The cycle $\Gamma$ is the union of
$\lambda$ paths, for some $\lambda\in \NN$, each visiting all lines in $\L$.
Similarly to \eqref{eq:Z5} and \eqref{eq:Z6}, the total length of the
segments of $\Gamma$ along the lines in $\L$ is at most
$6.77\, \lambda\, \len(C)$. Consequently, one of the $\lambda$ paths
$\gamma_1(q,H)\subset \Gamma$ satisfies $Z_1(q,H) \leq 6.77\, \len(C)$.
This path visits all lines in $\L$ and its length (including transversal
segments) is at most
$$ \len(\gamma_1(q,H)) \leq \left(3 + 6.77\right) \len(C)
\leq 9.77\, \len(C).$$

In both cases, $G_{\L}$ contains a group Steiner tree for $\L$ of length
at most $83\, \len(C)$, as required.
\end{proof}

\section{TSPN for unit disks and balls}  \label{sec:disks+balls}

In this section we prove Theorems~\ref{thm:3} and~\ref{thm:4} concerning
TSPN for unit disks and balls. Congruent disks are without a doubt among
the simplest neighborhoods~\cite{AH94,DM03}.
TSPN for unit disks is NP-hard, since when the disk centers are fixed and the radius
tends to zero, the problem reduces to a TSP for points. Given a set $S$ of $n$ points
in the plane, let $\D=\D(S,r)$ be the set of $n$ disks of radius $r$
centered at the points. It is known (and easy to argue) that the
optimal tours for the points and the disks, respectively, are
polygonal tours with at most $n$ sides.
The lengths of the optimal tours for the points and the
disks are not too far from each other. Indeed, given any tour of the
$n$ disks, one can convert it into a
tour of the $n$ centers by adding detours of length at most
$2r$ at each of the $n$ visiting points (arbitrarily selected);
see~\eg,~\cite{DM03,HHH11}.
Let $\opt(S)$ denote a shortest TSP tour
of $S$, and $\opt(S,r)$ denote a shortest TSP tour
of the disks of radius $r$ centered at the points in $S$.
Consequently, for each $n \geq 3$ and $r>0$, we have:
\begin{equation} \label{E1}
\len(\opt(S))- \len(\opt(S,r)) \leq 2 \, nr.
\end{equation}

As it is currently the case with TSP for points, the known approximation
schemes are highly impractical; see the comments in~\cite{MPS98}.
This is even more so for the approximation schemes for TSP with neighborhoods,
including disks, such as those in~\cite{BFG+09,DM03}. Designing more efficient
constant approximation algorithms remains of high interest. The obvious motivation
is to provide faster and conceptually simpler algorithmic solutions.

\subsection{Unit disks: an improved approximation}  \label{sec:disks}

\paragraph{Background.}
The current best approximation ratio for the TSP with $n$ unit disks,
$7.62$, was obtained in~\cite{DM03}. The algorithm works by reducing
the problem for $n$ disks to one for at most $n$ (representative)
points (representative points could be shared).
These points are selected after computing a \emph{line cover}
consisting of parallel lines. More generally, this ratio holds for
translates of a convex region.
An alternative approach (also from~\cite{DM03}) selects representative
points from among the centers of the disks (\ie, a suitable subset).
However, the approximation obtained in~\cite{DM03} in this way is weaker.
For instance, starting from a $(1+\eps)$-approximation for the center points
yields a ratio of $(8+\pi)(1+\eps) \leq 11.16$, provided that $\eps \leq 0.001$.
Starting from a $1.5$-approximation (with a faster algorithm) for the
center points yields a ratio of $(8+\pi) 1.5 \leq 16.72$.

Here we improve the two asymptotic approximation ratios,
from $7.62$ to $6.75$ (when using the PTAS for points),
and from $11.43$ to $8.52$ (when using the faster $1.5$-approximation
for points).
Somewhat surprisingly, we employ the latter approach with center points,
which gave previously only a weaker bound. It is worth mentioning that
the ratios for the special case of disjoint unit disks remain
unchanged, at $3.55$ and $5.32$, respectively.  We now proceed with
the details.

\paragraph{A simple packing argument.}
Let $B(x)$ denote a ball of radius $x$ centered at the origin.
Let $G=(V,E)$ be a connected geometric graph in $\RR^2$ and let
$L=\len(G)$. Let $C$ be the set of points at distance at most $x$ from
the edges and vertices of $G$. Equivalently, $C=G+B(x)$ is the
Minkowski sum of $G$ and $B(x)$. We need the following inequality;
see also~\cite[Lemma 4]{EFS09}.

\begin{lemma} \label{lem:minkowski1}
$\area(C) \leq 2Lx + \pi x^2$.
This bound cannot be improved.
\end{lemma}
\begin{proof}
Start by marking an arbitrary vertex $v_0$ of $G$;
The area covered by the Minkowski sum $B(x)+v_0$ is $\pi x^2$.
Pick an edge $uv$ of $G$ where $u$ is marked and $v$ is unmarked.
Place $B(x)$ with the center at $u$ and translate $B(x)$ along $uv$
(its center moves from $u$ to $v$), and mark $v$.
Observe that the newly covered area is at most
$2 |uv| x$. Continue and repeat this step as long as there
are unmarked vertices. Since $G$ is connected the procedure will
terminate when all vertices of $G$ are marked.
It follows that the area of $C$ is at most
$$ \pi x^2 + \sum_{uv \in E(G)} 2 |uv| x = 2Lx + \pi x^2, $$
as required.

Equality holds if and only if $G$ is a straight-line path.
Indeed, except for the first step (\ie, in each step involving an edge)
the newly covered area is strictly less than $2 |uv| x$, unless
all edges of $G$ are collinear in a straight-line path.
\end{proof}

\begin{figure}[htb]
\centerline{\epsfxsize=6.3in \epsffile{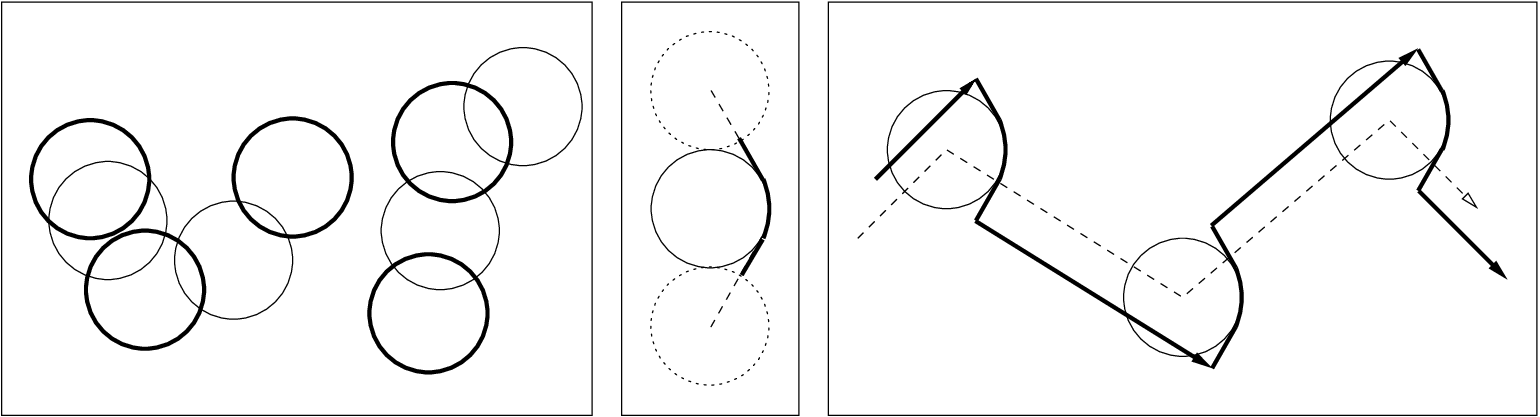}}
\caption{From left to right: (i) a line-sweep independent set (in bold
  lines); (ii) the curve $\gamma$; (iii) a part of the constructed disk tour.}
\label{f2}
\end{figure}

\paragraph{Approximation algorithm---outline.}
The idea is to first compute a maximal independent set and then an
approximate tour of the centers of the independent set, as in~\cite{DM03}.
The approximate tour of the centers is then extended by detours
so that it visits all the other disks (not in the independent set).
However the details differ significantly in both phases of the algorithm,
in order to obtain a better approximation ratio: a monotone
independent set is found, and a tailored visiting procedure is
employed that takes advantage of the special form of the independent set.

Let $\D$ be a set of unit disks.
First, compute a maximal independent set of disks $\I \subset \D$
by the following line-sweep algorithm. Select a leftmost disk $\omega \in \D$
and include it in $\I$. Remove from $\D$ all disks intersecting
$\omega$. Repeat this selection step as long as $\D$ is non-empty.

We call $\I$ a \emph{line-sweep independent set}
or \emph{$x$-monotone independent set}.
Clearly, $\I$ is a maximal independent set in $\D$, that is,
each disk in $\D \setminus \I$ intersects a disk in $\I$.
Moreover, by construction,
each disk in $\D \setminus \I$ intersects the right half-circle boundary
of a disk in $\I$.
Let $L^* =\len(\opt(\D))$ and $L_\I^* =\len(\opt(\I))$. Obviously,
$L_\I^* \leq L^*$.

\paragraph{Algorithm.}
The algorithm for computing a TSP tour of the disks is as follows.
Compute a (maximal) line-sweep independent set $\I$; write $k=|\I|$.
Next, compute $T_\I= o_1 \ldots o_k $, an $\alpha$-approximate tour of
the center points of disks in $\I$, for some constant $\alpha>1$.
If we use the PTAS for Euclidean TSP~\cite{Ar98,Mi99}, for a given
$0 <\eps < 1/2$, we have $\alpha=1+\eps$. If we use the approximation
algorithm for metric TSP due to Christofides~\cite{Ch76}, we have $\alpha=1.5$.

Write $S_\I =\{o_1,o_2,\ldots,o_k\}$.
For each disk $\omega \in \I$, let $\omega^-$ and $\omega^+$ be the
two unit disks tangent to $\omega$ from below and from above, respectively.
Let $o^-$ and $o^+$ be the centers of $\omega^-$ and
$\omega^+$, respectively. See Fig.~\ref{f2}(ii). Let
$\gamma(\omega)$ be the  open curve obtained as follows: start with
the tangent segment of positive slope from $o^-$ to $\omega$;
concatenate the arc of $\omega$ subtending a center angle of $\pi/3$
and symmetric about the $x$-axis; concatenate the tangent segment of
negative slope from $\omega$ to $o^+$.
Now remove two unit segments, one from each endpoint of the curve
obtained in the previous step.
The resulting curve is $\gamma=\gamma(\omega)$.
Observe that the open curve $\gamma(\omega)$ intersects any unit disk
from $\D$ that intersects the right half-circle boundary of $\omega$
(this includes $\omega$ as well). Let $v$ denote the vertical segment
connecting the endpoints of $\gamma$. It is easy to check that
\begin{align} \label{E6}
\len(\gamma) &=
2 \left( \frac{\pi}{6} + 2 \cos \frac{\pi}{6} -1 \right)=
2 \left( \frac{\pi}{6} + \sqrt3 -1 \right) \leq 2.512, \nonumber \\
\len(v) &= 4-\sqrt3 \leq 2.268.
\end{align}

Replace each segment $o_i o_{i+1}$ of this tour, with $i$ odd,
by a parallel segment of equal length connecting the two highest
endpoints of the curves $\gamma(\omega_i)$ and $\gamma(\omega_{i+1})$.
Similarly, replace each segment $o_i o_{i+1}$ of this tour, with
$i$ even, by a parallel segment of equal length connecting the two
lowest endpoints of $\gamma(\omega_i)$ and $\gamma(\omega_{i+1})$.
See Fig.~\ref{f2}(iii).

To obtain a tour (closed curve) we visit the disks in
$\I$ in the same order as $T_\I$. After each segment, the tour
traverses the corresponding curve $\gamma(\omega)$
(going up or down, as needed, in an alternating fashion).
If $k$ is even we proceed as above, while if $k$ is odd,
the curve $\gamma(\omega_1)$ is traversed in a circular way
(going down along $\gamma$ and up again along the vertical segment $v$)
in order to get a closed curve. We call $T$ the resulting tour.

\paragraph{Algorithm analysis.}
Since any disk in $\D$ is either in $\I$ or intersects the curve
$\gamma(\omega)$ of some disk $\omega \in \I$, and since
$T$ visits all disks in $\I$ and contains the curves $\gamma(\omega)$
of all disks in $\I$, it follows that $T$ is a valid tour for all disks in $\D$.
Further observe that the disjoint unit disks in $\I$ are contained in the
figure $C=T^*_\I + B(2)$. By Lemma~\ref{lem:minkowski1},
$ \pi |\I| \leq \area(C) \leq 4 \,\len(T^*_\I) +4\pi $, hence
\begin{equation} \label{E7}
k= |\I| \leq \frac{4}{\pi} L^*_\I +4 \leq \frac{4}{\pi} L^* +4.
\end{equation}

The total length of the detours incurred by $T$ over all disks in $\I$
is $k\, \len(\gamma)$ when $k$ is even, and  $k\, \len(\gamma)+ \len(v)$
when $k$ is odd. Hence by~\eqref{E6} the length of
the output tour is bounded from above as follows.
\begin{equation} \label{E8}
L \leq L_{S_\I} + k\, \len(\gamma)+ \len(v) \leq L_{S_\I} + (2.512 k + 2.268).
\end{equation}

Inequality~\eqref{E7} implies the following upper bound on the second
term in~\eqref{E8}.
\begin{equation} \label{E9}
2.512 k + 2.268 \leq 2.512 \left(\frac{4}{\pi} L^* +4 \right) + 2.268.
\end{equation}

We next bound from above the first term in~\eqref{E8}.
The inequality~\eqref{E1} applied to $\I$ and $S_\I$ yields
\begin{equation}
L^*_{S_\I} \leq L^*_\I + 2k.
\end{equation}

Since the algorithm computes a $\alpha$-approximation of the optimal
tour for the points in ${S_\I}$, by~\eqref{E7} we have
\begin{align} \label{E10}
L_{S_\I} &\leq \alpha L^*_{S_\I} \leq \alpha (L^*_\I + 2k)
\leq \alpha (L^*+ 2k) \nonumber \\
&\leq \alpha \left( L^*+ 2 \left(\frac{4}{\pi} L^* +4\right) \right)
\nonumber \\
&\leq \alpha \left( \left( 1+ \frac{8}{\pi} \right) L^* +8 \right).
\end{align}

Substituting into~\eqref{E8} the upper bounds in~\eqref{E10} and~\eqref{E9}
yields
\begin{align} \label{E11}
L &\leq \alpha \left( \left( 1+ \frac{8}{\pi} \right) L^* +8 \right)
+ 2.512 \left( \frac{4}{\pi} L^* +4 \right) + 2.268 \nonumber \\
&\leq \left( \alpha \left(1+\frac{8}{\pi} \right) +
2.512 \cdot \frac{4}{\pi} \right) L^* +
(8 \alpha + 4 \cdot 2.512 + 2.268) \nonumber \\
&\leq (3.5465 \alpha + 3.1984) L^* + (8 \alpha + 12.32).
\end{align}

For $\alpha=1+\eps$ (using the PTAS for the center points),
the length of the output tour is
$L \leq 6.75  \, L^* + 20.4$, assuming that $\eps \leq 0.001$.
A more precise calculation along the lines above yields the
following upper on the main term (in $L^*$); the constant factor
appears in Theorem~\ref{thm:3}; note also that $1/0.53 > 1.8$, which
explains the other parameter in Theorem~\ref{thm:3}.
$$ \left(\frac{7}{3} + \frac{8\sqrt3}{\pi}\right)
\left( 1+ \frac{\left(1+\frac{8}{\pi}\right) \eps}
{\left(\frac{7}{3} + \frac{8\sqrt3}{\pi}\right)} \right) L^* \leq
\left(\frac{7}{3} + \frac{8\sqrt3}{\pi}\right) (1+0.53 \, \eps) L^* . $$
The running time is dominated by that of computing
a $(1+\eps)$-approximation of the optimal tour of $n$ points in $\RR^2$.

For $\alpha=1.5$ (using the algorithm of Christofides for the center points),
the length of the output tour is
$L \leq 8.52  \, L^* + 24.4$.
The running time is dominated by that of computing a minimum-length
perfect matching on $n$ points in the plane ($n$ even), \eg,
$O(n^{3/2} \log^5{n})$ by using the algorithm of Varadarajan~\cite{Va98}.

\paragraph{Remarks.} 1. If the input consists of pairwise-disjoint (unit)
disks, then~\eqref{E10} yields improved approximations. These are not
new: the case $\alpha=1+\eps$ was already analyzed in~\cite{DM03};
we just list them for comparison.
For $\alpha=1+\eps$,~\eqref{E10} yields
$L \leq 3.55 \, L^* + 8.01$, assuming that $\eps \leq 0.001$.
For $\alpha=1.5$,~\eqref{E10} yields
$L \leq 5.32 \, L^* + 12$.
The approximation ratio $3.55$ for disjoint unit disks is probably far from
tight; the current best lower bound is $2$, see~\cite{DM03}.
The example in~\cite[Fig.~4]{HHH11} is yet another instance with a
ratio (lower bound) of $2$.
Hence the approximation ratio $6.75$ for unit disks (which uses the
above) is probably also far from tight.

\smallskip
2. A simple example shows that one cannot extend the
above approach to disks of arbitrary radii. Let $x \geq 1$.
See Fig.~\ref{f5}~(left)
where $n=3$, and  Fig.~\ref{f5}~(right) for its analogue with
arbitrarily large $n$. Let $x \to \infty$ and $\eps \to 0$.

(i) Suppose that we first compute a maximal independent set $\I$ in a greedy
manner, by selecting disks in increasing order of their radii.
Further suppose that we start by computing $T$, a constant approximation for the
shortest TSP tour on $\I$, for instance by using the algorithm of
de Berg~\etal~\cite{BGK+05}; recall, this algorithm works with fat,
disjoint regions. In some instances, no constant factor extension
(by adding suitable detours to visit the remaining disks)  exists.
In Fig.~\ref{f5}~(left), $\len(\opt(\I)) =2\eps$, while $\len(\opt)=4x$.
Moreover, since $x \to \infty$, no \emph{asymptotic} constant factor
can be guaranteed by this approach; indeed, for any constants
$\alpha,\beta$, there exists $x$ large enough, such that
$\alpha \, 2\eps + \beta < 4x$.
\begin{figure}[htb]
\centerline{\epsfxsize=3.9in \epsffile{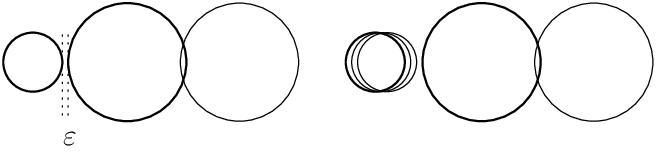}}
\caption{A set of three disks of radii $1$, $x$, and $x$,
  centered at $0$, $1+x+\eps$ and $1+3x$ (left) and a set of $n$
  disks, $n \geq 3$, of radii $1,\ldots,1$, $x$ and $x$ (right). A
  maximal independent set of disks (in bold) is shown for each case.}
\label{f5}
\end{figure}

(ii) Suppose that we first compute a maximal line-sweep independent
set, as in our algorithm for unit disks. The same example depicted
in Fig.~\ref{f5}~(left) shows that no constant factor extension
(by adding suitable detours to visit the remaining disks)  exists.
Moreover, as in (i), since $x \to \infty$, no \emph{asymptotic}
constant factor can be guaranteed by this approach.

\smallskip
3. Consider an algorithm that first computes a \emph{maximal}
independent set $\I$ of disks (according to some criterion), then computes
a good approximate tour of the disks in $\I$, and then extends this
tour with the boundary circles of the disks in $\I$ (in some way).
Observe that the length of the overall detour incurred in this way is
proportional to $\sum_{i \in \I} r_i$. The following claim (and
example) shows a deeper cause for which this general approach does not
give a constant approximation ratio; see also~\cite{DT14} for
refinements of this inequality and other related results.

\medskip
\emph{Claim. For every $M>0$, there exists a disk packing in the unit
  square $[0,1]^2$ with $\sum r_i \geq M$ and all disks tangent to the
  unit segment $[0,1] \times [0,0]$. }
\begin{proof}
We place disks in layers of decreasing radius. Each layer consists of
congruent disks placed in blocks in between consecutive tangent disks
of the previous layer, or in between a disk and a vertical side,
as in Fig.~\ref{f6}.
\begin{figure}[htb]
\centerline{\epsfxsize=3.1in \epsffile{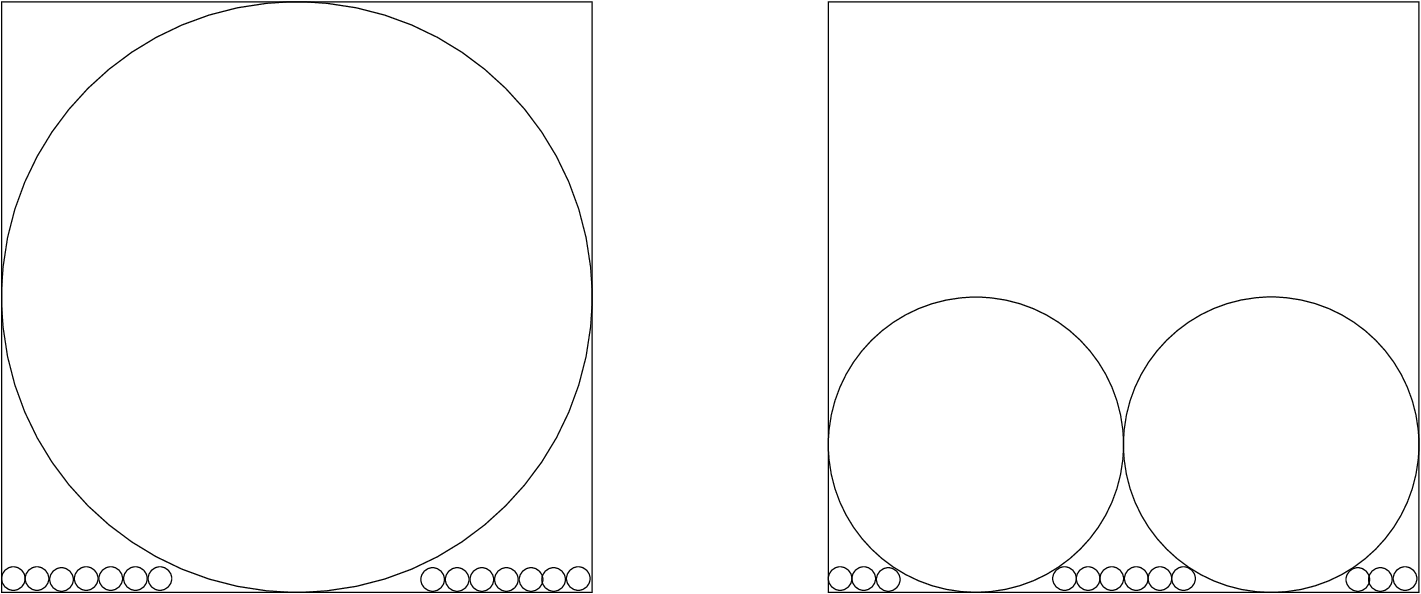}}
\caption{The first two layers of an iterative construction: $k=1$
  (left), and $k=2$ (right).}
\label{f6}
\end{figure}
The first layer consists of $k$ disks of radius $1/(2k)$, for some $k \geq 1$.
By choosing the radius of the disks in the next layer much smaller
than the radius of the disks in the current layer, one can ``cover'' any
prescribed large fraction $\rho<1$ of the length of the bottom side of
the square by disks tangent to the bottom side of the square and
having the sum of radii at least $\rho/2$. Consequently, by using
sufficiently many layers, one can achieve $\sum r_i \geq M$, as required.
\end{proof}

\subsection{Unit balls in $\RR^3$: an improved approximation}
\label{sec:balls}

We need an analogue of Lemma~\ref{lem:minkowski1}, specifically
Lemma~\ref{lem:minkowski2} below; its proof works in the same way.
Let $B(x)$ denote a ball of radius $x$.
Let $G=(V,E)$ be a connected geometric graph in $\RR^3$ and let
$L=\len(G)$. Let $C$ be the set of points at distance at most $x$ from
the edges and vertices of $G$.
Equivalently, $C=G+B(x)$ is the Minkowski sum of $G$ and $B(x)$.
\begin{lemma} \label{lem:minkowski2}
$ \vol(C) \leq \pi x^2 L + \frac{4\pi}{3} x^3$.
This bound cannot be improved.
\end{lemma}

Let $\D$ be a set of unit balls (as input). As in the planar case, we
compute a maximal independent set of disks $\I \subset \D$
by a plane-sweep algorithm. For convenience, we sweep a horizontal
plane in the positive direction of the $z$-axis.
We call $\I$ a \emph{plane-sweep independent set}
or \emph{$z$-monotone independent set}.

The algorithm computes a tour of $\D$ as follows.
First, compute a maximal $z$-monotone independent set $\I$; write $k=|\I|$.
Next, compute $T_\I= o_1 \ldots o_k $, an $\alpha$-approximate tour of
the center points of the balls in $\I$, for some constant $\alpha>1$.
Write $S_\I =\{o_1,o_2,\ldots,o_k\}$.
For each ball $\omega \in \I$, let $\Gamma=\Gamma(\omega)$ be a discrete set
of $28$ lattice points associated with $\omega$ (relative to its
center). For describing this set we will assume for convenience that
the center of $\omega$ is $(0,0,0)$. Let $a=1/\sqrt3$.
$\Gamma$~contains $16$ points in the plane $z=a$ and
$12$ points in the plane $z=3a$; see Fig.~\ref{f3}.
Specifically,
\begin{align*}
\Gamma &= \{(-3a,-3a,a), (-3a,-a,a), (-3a,a,a), (-3a,3a,a), \\
& (-a,-3a,a), (-a,-a,a), (-a,a,a), (-a,3a,a), \\
& (a,-3a,a), (a,-a,a), (a,a,a), (a,3a,a), \\
& (3a,-3a,a), (3a,-a,a), (3a,a,a), (3a,3a,a)\} \\
&\cup \{(-3a,-a,3a), (-3a,a,3a),
(-a,-3a,3a), (-a,-a,3a), (-a,a,3a), (-a,3a,3a), \\
& (a,-3a,3a), (a,-a,3a), (a,a,3a), (a,3a,3a),
(3a,-a,3a), (3a,a,3a)\}.
\end{align*}

One can check that the points in $\Gamma$ admit a Hamiltonian path
in which each edge has length $2a$,
say $\xi(\Gamma)=\gamma_1,\gamma_2,\ldots,\gamma_{28}$,
starting at $\gamma_1=(-a,-3a,a)$ and ending at $\gamma_{28}=(-a,-3a,3a)$.

\begin{figure}[htb]
\centerline{\epsfxsize=2.1in \epsffile{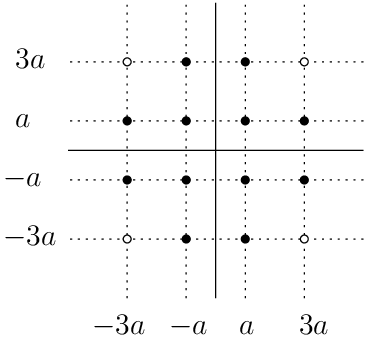}}
\caption{The set $\Gamma$ has $16$ points with $z=a$ and
$12$ points with $z=3a$; $|\Gamma|=28$.
The hollow circles indicate the four missing points in the
plane $z=3a$.}
\label{f3}
\end{figure}

We will prove shortly that any unit ball that intersects
$\omega$ from above (\ie, the $z$-coordinate of its center is
non-negative) contains at least one of the points in $\Gamma(\omega)$.
Moreover, this also holds for $\omega$ itself.

We modify (extend) the tour $T_\I= o_1 \ldots o_k $ as follows.
Assume first that $k$ is even.
We replace each segment $o_i o_{i+1}$ of this tour, with $i$ odd,
by a parallel segment of equal length connecting
$\gamma_1 \in \Gamma(\omega_i)$ with
$\gamma_1 \in \Gamma(\omega_{i+1})$.
Similarly, we replace each segment $o_i o_{i+1}$ of this tour, with
$i$ even, by a parallel segment of equal length connecting
$\gamma_{28} \in \Gamma(\omega_i)$ with
$\gamma_{28} \in \Gamma(\omega_{i+1})$.
To obtain a tour, we visit the balls in
$\I$ in the same order as $T_\I$. After each segment, the tour
visits all the $28$ points in the corresponding set $\Gamma(\omega)$
by using the Hamiltonian path $\xi(\Gamma)$
and then continues with the next segment, etc.
This extension procedure can be adapted to work for odd $k$
\emph{without} incurring any increase in cost: specifically,
the first cycle of period $2$ is replaced by a cycle of period $3$.
For odd $k$, the output TSP tour has the form
$T=\xi_1 \xi_2 \xi_3 \xi \xi^R \xi \xi^R \ldots \xi \xi^R$,
rather than the form $T=\xi \xi^R \xi \xi^R \ldots \xi \xi^R$ (for $k$ even).
Here $\xi^R$ is the path $\xi$ traversed in the opposite direction, and
$\xi_1,\xi_2,\xi_3$ are three suitable Hamiltonian paths on $\Gamma$
(details are omitted).

\paragraph{Algorithm analysis.}
The analysis of the approximation ratio is similar to that in the
planar case. The disjoint unit balls in $\I$ are contained in the
body $C=T^*_\I + B(2)$. By Lemma~\ref{lem:minkowski2},
$$ \frac{4\pi}{3} |\I| \leq \vol(C)
\leq 4\pi \,\len(T^*_\I) + \frac{4\pi}{3} \, 8, $$
hence
\begin{equation} \label{E17}
k= |\I| \leq 3 \left(L^* + \frac83 \right) = 3 L^* +8.
\end{equation}

The total length of the detours incurred by $T$ over all
balls in $\I$ is bounded from above by
\begin{equation} \label{E18}
(28-1) 2a k = 27 \frac{2}{\sqrt3} k = 18 \sqrt3 k.
\end{equation}

It follows that the length of the output tour is bounded from above as
follows.
\begin{equation} \label{E19}
L \leq L_{S_\I} + 18 \sqrt3 k.
\end{equation}

The upper bound on $L_{S_\I}$ (analogue of~\eqref{E10}) is
\begin{align} \label{E20}
L_{S_\I} &\leq \alpha L^*_{S_\I} \leq \alpha (L^*_\I + 2k)
\leq \alpha (L^*+ 2k)
\leq \alpha ( L^*+ 2 (3 L^* +8) )  \nonumber \\
&= 7 \alpha L^* + 16 \alpha.
\end{align}

The upper bound on $18 \sqrt3 k$ (analogue of~\eqref{E9}) is
\begin{equation} \label{E21}
18 \sqrt3 k \leq 18 \sqrt3 (3 L^* +8) = 54 \sqrt3 L^* + 144 \sqrt3.
\end{equation}

Substituting into~\eqref{E19} the upper bounds in~\eqref{E20} and~\eqref{E21}
yields
\begin{align} \label{E22}
L &\leq (7 \alpha L^* + 16 \alpha) + (54 \sqrt3 L^* + 144 \sqrt3) \nonumber \\
&= (7 \alpha + 54 \sqrt3)  L^* + (16 \alpha + 144 \sqrt3).
\end{align}

For $\alpha=1+\eps$ (using the PTAS for the center points),
the length of the output tour is
$L \leq 100.61  \, L^* + 265.6$, assuming that $\eps \leq 0.01$.
For $\alpha=1.5$ (using the algorithm of Christofides for the center points),
the length of the output tour is
$L \leq 104.1  \, L^* + 273.5$.
The running time is dominated by that of computing a minimum-length
perfect matching on $n$ points in $\RR^3$ ($n$ even), \eg, $O(n^{3})$~\cite{Ga76}.

\begin{lemma} \label{L5}
Let $\omega$ and $\omega'$ be two intersecting unit balls, centered at
$(0,0,0)$ and $(x,y,z)$, respectively, where $z \geq 0$.
Then $\omega$ contains a point in $\Gamma(\omega)$.
\end{lemma}
\begin{proof}
By symmetry, it suffices to prove the claim when
$x,y \geq 0$. We therefore have $x,y,z \geq 0$ and
$x^2 + y^2 +z^2 \leq 4$. We distinguish two cases, depending on
whether $z \leq 2a$ or $z \geq 2a$. If  $z \leq 2a$, we show that
$\omega$ contains a point of $\Gamma$ in the lower plane
$\sigma_1 \ : \ z=a$;
if $z \geq 2a$, we show that $\omega$ contains a point of $\Gamma$ in
the higher plane $\sigma_3 \ : \ z=3a$.
Write $\Gamma_1 = \Gamma \cap \sigma_1$, and
$\Gamma_3 = \Gamma \cap \sigma_3$.

\smallskip
\emph{Case 1:}  $z \leq 2a$. Since $x^2 + y^2 +z^2 \leq 4$, we have
$\max (x,y) \leq 2 < 4a$. The closest lattice point
$\gamma=(\gamma_x,\gamma_y,\gamma_z) \in \Gamma_1$
to $(x,y,z)$ satisfies
$$ |x-\gamma_x| \leq a, \ \ |y-\gamma_y| \leq a, \ \
{\rm and} \ \  |z-\gamma_z| \leq a, $$
thus
$$ (x-\gamma_x)^2 + (y-\gamma_y)^2 + (z-\gamma_z)^2 \leq
3a^2=1, $$
as required.

\smallskip
\emph{Case 2:}  $z \geq 2a$. Since $x^2 + y^2 +z^2 \leq 4$, we have
$x^2 +y^2 \leq 4-4a^2 =8/3$. Observe that the disk $x^2 +y^2 \leq 8/3$
does not intersect the interior of the square $[2a,3a]^2$
in the plane $z=0$. Thus the projection of $(x,y,z)$ onto the plane $z=0$
is contained in $[0,3a]^2 \setminus (2a,3a]^2$. This implies that
the closest lattice point
$\gamma=(\gamma_x,\gamma_y,\gamma_z) \in \Gamma_3$
to $(x,y,z)$ satisfies
$$ |x-\gamma_x| \leq a, \ \ |y-\gamma_y| \leq a, \ \
{\rm and} \ \  |z-\gamma_z| \leq a, $$
and the conclusion follows as in Case 1.
\end{proof}

\paragraph{Remark.} Analogous to the planar case, if the input
consists of pairwise-disjoint (unit) balls, then~\eqref{E20} yields
improved approximations.
For $\alpha=1+\eps$,~\eqref{E20} yields
$L \leq 7.01\, L^* + 16.1$, assuming that $\eps \leq 0.001$.
For $\alpha=1.5$,~\eqref{E20} yields
$L \leq 10.5 \, L^* + 24$.

\paragraph{Generalization to higher dimensions.}
The technique in this section generalizes to congruent balls in $\RR^d$ for
any fixed $d\geq 4$. First, the plane-sweep algorithm does so and yields an
independent set $\I$. Then compute an $\alpha$-approximate tour $T_\I$ of the center
points of the balls in $\I$ for a small $\alpha \leq 1.5$.

\begin{figure}[htb]
\centerline{\epsfxsize=3.2in \epsffile{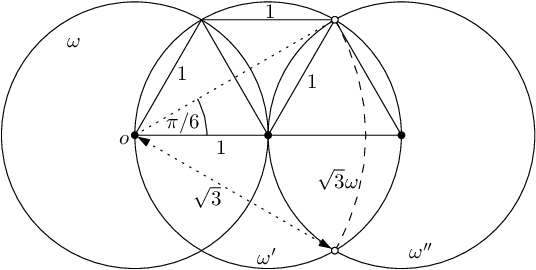}}
\caption{A unit disk $\omega$ centered at $o$ intersects two unit disks,
  $\omega'$ and $\omega''$, whose centers are at distance 1 and 2 from $o$.
  Both $\omega'$ and $\omega''$ intersects the boundary of $\sqrt{3}\omega$
  in a spherical cap of radius $\sqrt{3}\cdot \pi/6$.}
\label{fig:sphericalcap}
\end{figure}

For each ball $\omega \in \I$, we construct a finite point set
$\Gamma=\Gamma(\omega)$ with the property that any unit ball that
intersects $\omega$ contains at least one of the points in $\Gamma(\omega)$.
Consider a unit ball $\omega'$ that intersects $\omega$. If the distance between
their centers is less than 1, then $\omega'$ contains the center of $\omega$;
otherwise $\omega'$ intersects the boundary of $\sqrt{3}\, \omega$ (i.e., the ball of
radius $\sqrt{3}$ concentric with $\omega$) in a spherical cap of radius at least
$\sqrt{3}\, \frac{\pi}{6}$ in spherical distance (refer to Fig.~\ref{fig:sphericalcap}).
The bound $\sqrt{3}\, \frac{\pi}{6}$ is attained when the centers of $\omega$ and $\omega'$
are at distance 1 or 2 apart. Compute a maximal packing of the sphere
$\partial (\sqrt{3} \omega)$ with spherical caps of radius $\sqrt{3}\,
\frac{\pi}{12}$, starting with an arbitrary cap, and incrementally
adding interior-disjoint caps so that each touches some previous cap.

Let $\Gamma(\omega)$ contain the centers of all caps in this maximal packing
and the center of $\omega$. Suppose a unit disk $\omega'$ intersects $\omega$ but misses
$\Gamma(\omega)$. Then $\omega'$ contains a spherical cap in $\partial (\sqrt{3}\omega)$
of radius at least $\sqrt{3}\cdot \pi/6$, which contains no point in $\Gamma(\omega)$;
consequently a spherical cap with the same center and radius $\sqrt{3}\, \frac{\pi}{12}$
is disjoint from all caps in the packing, contradicting maximality.
Therefore $\Gamma(\omega)$ has the desired property.

We extend the tour $T_\I$ by suitable detours visiting all points in $\Gamma(\omega)$
for all $\omega \in \I$ and thereby obtain a tour for the input set.
The analysis of the approximation ratio is similar to the 2- and 3-dimensional cases
and uses volume arguments in $\RR^d$. Let $\vol_d(r)$ be the volume of
a ball of radius $r$ in $\RR^d$. It is well-known that
\begin{equation} \label{E26}
\vol_d(r)= \begin{cases}
\dfrac{\pi^{d/2}}{(d/2)!} \cdot r^d & {\rm if \ } d \ {\rm is \ even},
\medskip \\
\dfrac{2^d \cdot \pi^{(d-1)/2} \, ((d-1)/2)!}{d!} \cdot r^d & {\rm if \ }
d \ {\rm is \ odd}.
\end{cases}
\end{equation}

Combining~\eqref{E26} with the Stirling formula yields the following
upper bound:

\begin{lemma}\label{lem:stirling}
$$ \frac{\vol_{d-1}(1)}{\vol_d(1)} \leq (1+o(1)) \sqrt{\frac{d}{2\pi}}. $$
\end{lemma}
\begin{proof}
Write $f \sim g$ whenever $\lim_{d \to\infty} f(d)/g(d)=1$.
We distinguish two cases according to the parity of $d$.

If $d$ is even, then
\begin{eqnarray*}
\frac{{\rm Vol}_{d-1}(1)}{{\rm Vol}_d(1)}
&=& \frac{2^{d-1} \pi^{(d-2)/2} ((d-2)/2)!}{(d-1)!} \frac{(d/2)!}{\pi^{d/2}}\\
&=& \frac{2^d}{\pi} \frac{\pi^{d/2} (d/2)!}{d!} \frac{(d/2)!}{\pi^{d/2}}
\sim \frac{2^d}{\pi} \frac{(2 \pi d/2) \left( \frac{d}{2e} \right)^d}
{\sqrt{2\pi d} \left( \frac{d}{e} \right)^d} \\
&=& \frac{2^d}{\pi} \frac{\pi d}{\sqrt{2\pi d}} \frac{1}{2^d} = \sqrt{\frac{d}{2\pi}}.
\end{eqnarray*}

If $d$ is odd, then
\begin{eqnarray*}
\frac{{\rm Vol}_{d-1}(1)}{{\rm Vol}_d(1)}
&=& \frac{\pi^{(d-1)/2}}{((d-1)/2)!} \, \frac{d!}{2^d \pi^{(d-1)/2} \, ((d-1)/2)!}\\
&=& \frac{d!}{2^d ((d-1)/2)! ((d-1)/2)!}
\sim \frac{\sqrt{2\pi d} \left( \frac{d}{e} \right)^d}
{2^d \, 2 \pi \frac{d-1}{2} \left( \frac{d-1}{2e} \right)^{d-1}} \\
&=& \frac{\sqrt{2\pi d} \, d^d \, 2^{d-1} \, e^{d-1}}{\pi \, e^d \,
  2^d \, (d-1)^d}
\sim \frac{\sqrt{2\pi d}}{2 e \pi} \, e = \sqrt{\frac{d}{2\pi}}.
\end{eqnarray*}
\vskip-17pt
\end{proof}

By Lemma~\ref{lem:stirling}, a volume argument analogous to~\eqref{E17} yields
$$ k =|\I| \leq \frac{\vol_{d-1}(2) \, L^*+\vol_d(2)}{\vol_d(1)}
\leq  (1+o(1)) \sqrt{\frac{d}{2\pi}} \, 2^{d-1} L^*+2^d. $$

The surface area of a sphere of radius $r$ in $\RR^d$
is $\area_{d-1}(r)=2\pi r \vol_{d-2}(r)$,
and the surface area of a spherical cap of radius $r\varphi$ is
bounded from below by $\vol_{d-1}(r\sin\varphi)$. A volume argument yields
\begin{equation}\label{eq:tau}
|\Gamma|
\leq \frac{\area_{d-1}(\sqrt{3})}{\vol_{d-1}(\sqrt{3}\sin(\pi/12))} +1
\leq \frac{2\pi \vol_{d-2}(1)}{(\sin(\pi/12))^{d-1}\vol_{d-1}(1)} +1
\leq (1+o(1)) \frac{\sqrt{2 \pi d}}{(\sin(\pi/12))^{d-1}}.
\end{equation}
If two spherical caps of radius $\sqrt{3}\, \frac{\pi}{12}$ are in
contact on the sphere $\partial (\sqrt{3}\omega)$, then the distance
between their centers is $2 \sqrt{3} \sin\frac{\pi}{12}$.
By construction, the length of a minimum spanning tree of $\Gamma$ is
$$ (|\Gamma|-2) \, 2 \sqrt{3} \sin \frac{\pi}{12} +\sqrt{3}\leq
(1+o(1)) \frac{2 \sqrt{6\pi d}}{(\sin(\pi/12))^{d-2}}, $$
and the length of a Hamiltonian cycle $\xi$ of $\Gamma$ is at most
twice this length. Consequently, we obtain a tour of length
$$ L\leq \alpha L^*+2k\, \len(\xi)
\leq \alpha L^*+ 2\left((1+o(1)) \sqrt{\frac{d}{2 \pi}} \, 2^{d-1} L^*+2^d\right)
\left((1+o(1)) \frac{2 \sqrt{6\pi d}}{(\sin(\pi/12))^{d-2}} \right). $$

The resulting (asymptotic) approximation ratio is
$$ \alpha +(1+o(1)) \frac{2 \sqrt{3}\, d\, 2^{d}}{(\sin(\pi/12))^{d-2}}
= O \left( d \left( \frac{2}{\sin(\pi/12)} \right)^d \right)
= O \left( 7.73^d \right), $$
as claimed.

\section{Conclusion}  \label{sec:conclusion}

We revisited TSP with neighborhoods and obtained several
approximation algorithms: some for neighborhoods previously less
studied, such as lines and hyperplanes in $\RR^d$, and some
for the most previously studied, such as disks and balls.
Despite the progress, one may rightfully say that
the general problem of TSP with neighborhoods is far from resolved.
Interesting questions remain open regarding the structure of
optimal TSPN tours for lines, segments, balls, and hyperplanes, and the
degree of approximation achievable for these problems.
We record the simplest and most natural open questions on TSPN
that we could identify.

\begin{itemize}\itemsep -1pt

\item [(1)] Is there a polynomial-time exact algorithm for planes in $\RR^3$?

\item [(2)] Is there a constant approximation algorithm
for lines in $\RR^3$ (or in $\RR^d$ for $d\geq 3$)?
Can the current $O(\log^3 n)$ ratio be improved?

\item [(3)] Is there a constant approximation algorithm for planar convex bodies?

\item [(4)] Is there a constant approximation algorithm for parallel
segments in $\RR^3$? To start with, one can further assume that
the segments are pairwise-disjoint.

\item [(5)] Is there a constant approximation algorithm for balls (of arbitrary radii)
  in $\RR^3$?

\end{itemize}

\end{document}